\pdfoutput=1
\documentclass[aps,pra,amsmath,amssymb,twocolumn,10pt]{revtex4-2}
\usepackage{mathrsfs}
\usepackage[T1]{fontenc}
\fontencoding{T1}
\usepackage[utf8]{inputenc}
\usepackage{enumitem}
\usepackage{tikz}
\usepackage[normalem]{ulem}

\newcommand{\vertiii}[1]{{\left\vert\kern-0.25ex\left\vert\kern-0.25ex\left\vert #1 
    \right\vert\kern-0.25ex\right\vert\kern-0.25ex\right\vert}}

\usepackage{amsfonts}
\usepackage{amsmath,amssymb,amsthm}
\usepackage{graphicx}
\usepackage{fancyhdr}
\usepackage[breakable]{tcolorbox}
\usepackage{bbm}
\usepackage{url}
\usepackage{mathtools}
\usepackage{braket}
\usepackage{upgreek}
\usepackage{dsfont}
\usepackage{collectbox}
\usepackage{quantikz}

\usepackage[plain]{fancyref} 
\usepackage[ruled]{algorithm2e}

\usepackage{algpseudocode}

\usepackage{hyperref}

\newtheorem{thm}{Theorem}

\newtheorem{lem}{Lemma}
\newtheorem*{thm*}{Theorem}
\makeatletter
\newcommand{\setthmtag}[1]{
  \let\oldthethm\thethm
  \newcommand{\thethm}{#1}
  \g@addto@macro\endthm{
    \addtocounter{thm}{-1}
    \global\let\thethm\oldthethm}
  }
\makeatother

\newtheorem{prop}[thm]{Proposition}
\newtheorem*{prop*}{Proposition}
\newtheorem{lemma}[thm]{Lemma}
\newtheorem*{lemma*}{Lemma}
\newtheorem{cor}[thm]{Corollary}
\newtheorem*{cor*}{Corollary}

\newtheorem*{cj*}{Conjecture}
\newtheorem{Def}[thm]{Definition}
\newtheorem*{Def*}{Definition}

\theoremstyle{definition}

\newtheorem*{rem*}{Remark}

\def\beq{\begin{equation}}
\def\eeq{\end{equation}}
\def\bq{\begin{quote}}
\def\eq{\end{quote}}
\def\ben{\begin{enumerate}}
\def\een{\end{enumerate}}
\def\bit{\begin{itemize}}
\def\eit{\end{itemize}}

\def\lb{\left(}
\def\rb{\right)}

\def\r|{\right|}

\def\i{\operatorname{id}}

\newcommand\cM{\mathcal{M}}

\newcommand\cA{\mathcal{A}}
\newcommand\cH{\mathcal{H}}
\newcommand\cB{\mathcal{B}}
\newcommand\cL{\mathcal{L}}
\newcommand\cD{\mathcal{D}}

\newcommand{\U}{\mathcal{U}}

\newcommand{\Em}{\mathcal{E}}

\newcommand{\Lm}{\mathcal{L}}

\newcommand{\tr}[1]{\operatorname{Tr}\lb#1\rb}

\newcommand{\norm}[1]{\left\|#1\right\|}

\newcommand{\cO}{\mathcal{O}}

\newcommand{\cU}{\mathcal{U}}

\newcommand\be{\begin{equation}}
\newcommand\ee{\end{equation}}

\definecolor{ms}{rgb}{0,.4,1}

\newcommand{\Norm}[1]{\left\Vert #1 \right\Vert}

\def\r{\textbf{r}}

\makeatletter
\newcommand*\bigcdot{\mathpalette\bigcdot@{.6}}
\newcommand*\bigcdot@[2]{\mathbin{\vcenter{\hbox{\scalebox{#2}{$\m@th#1\bullet$}}}}}
\makeatother

\newcommand{\bigo}[1]{\mathcal{O}\left (#1\right)}

\newcommand{\rom}[1]{\uppercase\expandafter{\romannumeral #1\relax}}
\newcommand{\norbra}[1]{\left( #1\right)}
\newcommand{\sqrbra}[1]{\left[ #1\right]}

\newcommand{\lind}{\mathcal{L}}

\usepackage{amssymb}             

\newcommand{\de}{{\rm d}}

\newcommand{\E}{\mathcal{E}}

\begin{document}
\title{Polynomial-time thermalization and Gibbs sampling from system-bath couplings
}

\author{\begingroup
\hypersetup{urlcolor=navyblue}
\href{https://orcid.org/0000-0001-9699-5994}{Sam Slezak }
\endgroup}
\email[Samuel Slezak ]{samuel.slezak@ens-lyon.fr}
\affiliation{Univ Lyon, ENS Lyon, UCBL, CNRS, Inria, LIP, F-69342, Lyon Cedex 07, France}

\author{\begingroup
\hypersetup{urlcolor=navyblue}
\href{https://orcid.org/0000-0002-6395-3971}{Matteo Scandi }
\endgroup}
\email[Matteo Scandi ]{matteo.scandi@csic.es}
\affiliation{Instituto de F\'{i}sica T\'{e}orica UAM/CSIC, C. Nicol\'{a}s Cabrera 13-15, Cantoblanco, 28049 Madrid, Spain}

\author{\begingroup
\hypersetup{urlcolor=navyblue}
\href{https://orcid.org/0000-0001-9699-5994}{Daniel Stilck Fran\c{c}a }
\endgroup}
\email[Daniel Stilck Fran\c ca ]{dsfranca@math.ku.dk}
\affiliation{Department of Mathematical Sciences, University of Copenhagen, Universitetsparken 5, 2100 Denmark}

\author{\begingroup
\hypersetup{urlcolor=navyblue}
\href{https://orcid.org/0000-0002-5889-4022}{\'Alvaro M. Alhambra
\endgroup}
}
\email[\'{A}lvaro M. Alhambra ]{alvaro.alhambra@csic.es}
 \affiliation{Instituto de F\'{i}sica T\'{e}orica UAM/CSIC, C. Nicol\'{a}s Cabrera 13-15, Cantoblanco, 28049 Madrid, Spain}

 \author{\begingroup
\hypersetup{urlcolor=navyblue}
\href{https://orcid.org/0000-0001-7712-6582}{Cambyse Rouz\'{e}
\endgroup}
}
\email[Cambyse Rouz\'{e} ]{cambyse.rouze@inria.fr}
 \affiliation{Inria, T\'{e}l\'{e}com Paris - LTCI, Institut Polytechnique de Paris, 91120 Palaiseau, France}

\begin{abstract}
Many physical phenomena, including thermalization in open quantum systems and quantum Gibbs sampling, are modeled by Lindbladians approximating a system weakly coupled to a bath. Understanding the convergence speed of these Lindbladians to their steady states is crucial for bounding algorithmic runtimes and thermalization timescales. We study two such families of processes: one characterizing a repeated-interaction Gibbs sampling algorithm, and another modeling open many-body quantum thermalization. We prove that both converge in polynomial time for several non-commuting systems, including high-temperature local lattices, weakly interacting fermions, and 1D spin chains. These results demonstrate that simple dissipative quantum algorithms can prepare complex Gibbs states and that Lindblad dynamics accurately capture thermal relaxation. Our proofs rely on a novel technical result that extrapolates spectral gap lower bounds from quasi-local Lindbladians to the non-local generators governing these dynamics.
\end{abstract}

\maketitle

\section{Introduction}\label{sec:intro}
\noindent Systems in nature often \textit{thermalize}, converging to an equilibrium state that is independent of the microscopic details of the initial conditions. This process is typically understood as a system which is weakly coupled to a heat bath, which in the quantum realm is often described through a Lindbladian. Many previous works have attempted to mimic this thermalization process with quantum computers \cite{kastoryano2023quantum,chen2023efficient,Ding_2024_single,ding2024efficientquantumgibbssamplers,gilyén2024quantumgeneralizationsglaubermetropolis,hagan2025thermodynamiccostignorancethermal,jiang2024quantummetropolissamplingweak,Lambert_2024Fixing,lloyd2025quasiparticlecoolingalgorithmsquantum,Rall_2023,Temme_2011,Verstraete2009,zhan2025rapidquantumgroundstate}, with techniques falling under the umbrella term of \textit{quantum Gibbs sampling}. While the topic has a long history, recently relevant progress has been made in our understanding of both this family of processes and their quantum simulation, especially in the many-body context. In particular, the works of \cite{chen2023efficient,ding2024efficientquantumgibbssamplers} introduced a method to simulate Lindblad evolutions that both have an efficient implementation, and also satisfy Kubo-Martin-Schwinger (KMS) detailed balance exactly, which guarantees convergence to the Gibbs state. Later work characterized the mixing time of these Lindbladians in a variety of situations including high temperatures, weak interactions, and $1$D systems \cite{rouzé2024efficientthermalizationuniversalquantum,rouzé2024optimalquantumalgorithmgibbs,tong2025fastmixingweaklyinteracting,zhan2025rapidquantumgroundstate,smid2025rapidmixingquantumgibbs,bergamaschi2025quantumspinchainsthermalize}.  

The algorithms in \cite{chen2023efficient,ding2024efficientquantumgibbssamplers}, while conceptually elegant and asymptotically optimal, have complicated implementations including block encodings and operator Fourier transforms, and are not suitable for near term implementation.
In light of this, simpler Gibbs sampling algorithms with implementations much closer to the physical models they imitate were introduced in \cite{ding2025endtoendefficientquantumthermal,hahn2025provably,lloyd2025quantumthermalstatepreparation,hagan2025thermodynamiccostignorancethermal,Lambert_2024Fixing,lloyd2025quasiparticlecoolingalgorithmsquantum,Ramon_Escandell_2025}, putting us closer to the goal of generating complex equilibrium states in a quantum computer. These schemes are akin to the so-called \emph{collision} or \emph{repeated interaction} models \cite{Cattaneo_2021,Ciccarello_2022,Pocrnic_2025}. 

In particular, the algorithm introduced in \cite{ding2025endtoendefficientquantumthermal}, consisting of weak, repeated interactions with thermal qubits, showed end-to-end performance guarantees for restricted situations, such as commuting Hamiltonians. In our first contribution, we prove a collection of efficient preparation results for this algorithm, extending the performance guarantees to all the aforementioned fast mixing systems, which includes a wide range of physically relevant many-body settings \cite{rouzé2024efficientthermalizationuniversalquantum,rouzé2024optimalquantumalgorithmgibbs,tong2025fastmixingweaklyinteracting,smid2025rapidmixingquantumgibbs,bergamaschi2025quantumspinchainsthermalize}. This shows that a simple, early fault tolerant scheme is able to efficiently produce a variety of complex quantum Gibbs states. 

The notion of KMS detailed balance is relevant beyond quantum algorithms, since it also seems to be the correct notion of open-system thermalization towards the many-body Gibbs state \cite{Nathan_2020,Mozgunov_2020,scandi2025thermalizationopenmanybodysystems} (see \cite{Fazio_2025,stefanini2025lindbladme} for reviews). We also turn to the model of thermalization satisfying KMS detailed balance derived in \cite{scandi2025thermalizationopenmanybodysystems} from a weak system-bath interaction under natural open-system assumptions. We show that this evolution converges to its equilibrium state in polynomial time for two related settings: arbitrary local Hamiltonians at high temperatures \cite{rouzé2024optimalquantumalgorithmgibbs}, and deviations from product Hamiltonians \cite{smid2025rapidmixingquantumgibbs}. This proves from first principles that open many body systems weakly coupled to high-temperature baths reach their Gibbs steady state both quickly and accurately, as expected on physical grounds. Thus, KMS Lindbladians accurately capture the natural many-body thermalization, down to the steady state. Previous results on physically-generated Lindbladians were only known for commuting models \cite{kastoryano2016quantumgibbssamplerscommuting,Rapid1D,kochanowski2024rapidthermalizationdissipativemanybody}, and their associated Davies generators in the limit of vanishing coupling with the bath. As an additional result in this direction for commuting Hamiltonians, we use recently proven \emph{modified logarithmic Sobolev inequality} results to strengthen arguments on the thermalization times of quantum memories \cite{stengele2025modified}.

Our results on KMS Lindbladians are enabled by the same underlying technical contribution, which allows us to extrapolate spectral gap lower bounds from quasi-local to non-local Lindbladians, which we believe is of independent interest. Spectral gaps are well-known to govern the speed at which a dissipative evolution convergences to its equilibrium. The majority of currently known  proofs for spectral gaps of KMS symmetric Lindbladians rely on the Lindbladians consisting of terms with quasi-local support, allowing for an extrapolation to a trivially fast-mixing point. Both the algorithm in \cite{ding2025endtoendefficientquantumthermal} and the model in \cite{scandi2025thermalizationopenmanybodysystems} are precise only when they are too non-local for existing proof techniques to be applicable. Our lemma allows us to overcome this difficulty. Furthermore, we believe this technique can be used to show that many other Lindbladians out of the reach of current proof techniques are fast mixing as well and it allows us to conclude that various Gibbs samplers proposed in the literature have a qualitatively similar behaviour when it comes to their mixing times.

Overall, our results give theoretical footing to the early-fault tolerant dissipative preparation of complex quantum states, which has recently seen relevant experimental progress \cite{Mi_2024}
, and appears to be a promising avenue for practical quantum advantage \cite{Lin_2025}. 


\section{Background}

\begin{figure}[t]
    \centering
    \includegraphics[width=0.475\textwidth]{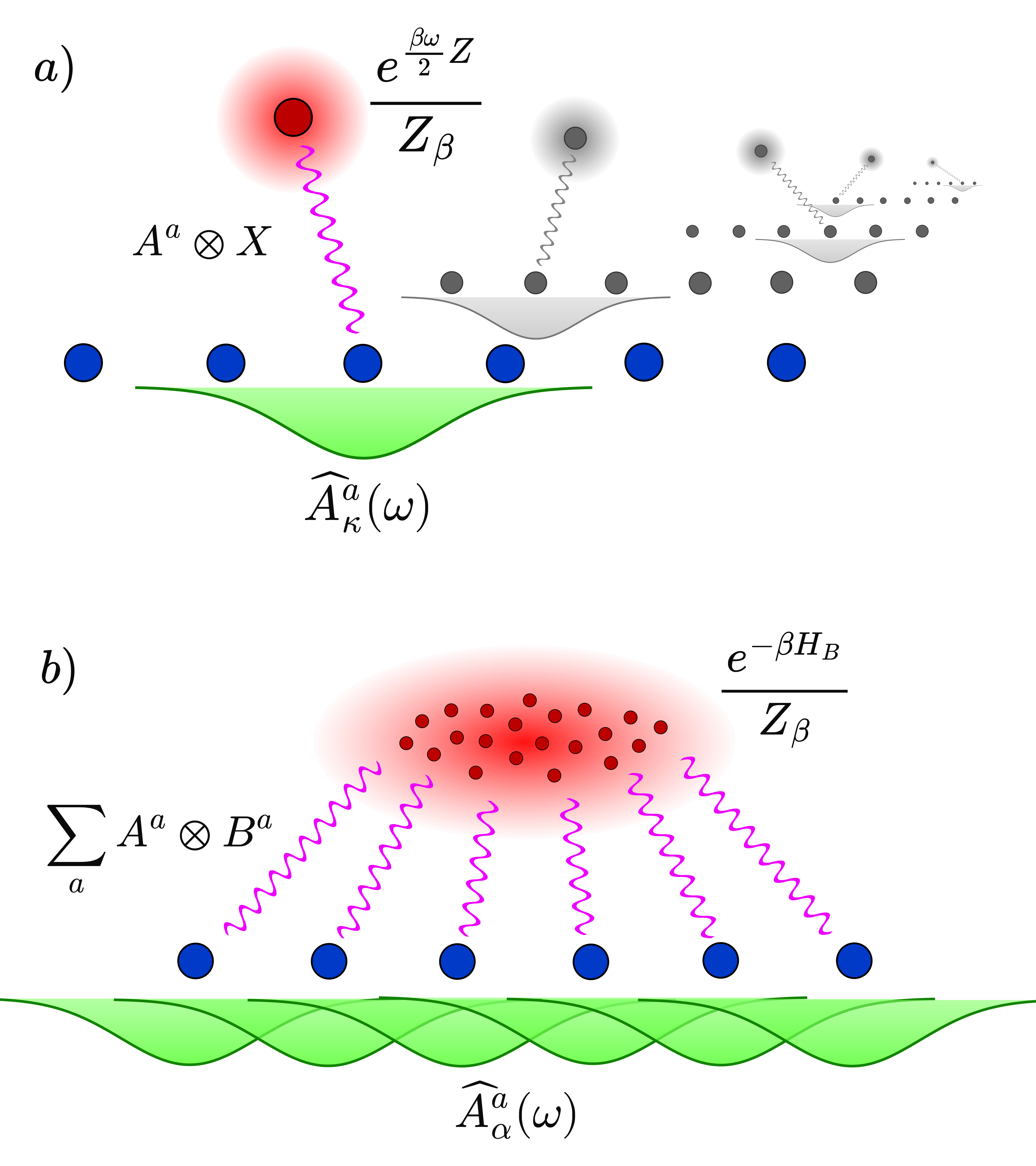}
    \caption{Representations of the system-bath couplings we consider. In a), each step of the algorithm couples a single qubit bath with a random frequency $\omega$ to a randomly chosen qubit with the interaction $V=A^a\otimes X$ with time dependent strength $\alpha f_\kappa(t)$, resulting in the jump operators $\widehat{A}_\kappa^a(\omega)$ in the Lindbladian with their locality is determined by the width of the function $f_\kappa$. In b) a bath of qubits is coupled to the system via the interaction $V=\sum_a A^a\otimes B^a$ with strength $\alpha$, leading to jump operators $\widehat{A}^a_\alpha(\omega)$ with their locality being determined by the bath correlation functions and the coupling strength $\alpha$.}
    \label{fig:system-bath-couplings}
\end{figure}
\noindent  We consider Lindbladians $\Lm$ which are self-adjoint with respect to the KMS inner product in the Heisenberg picture; that is:
\begin{align*}
    \langle Y, \Lm^\dagger(X)\rangle_{\rho_\beta} = \langle \Lm^\dagger(Y), X\rangle_{\rho_\beta}\;\forall\;X,Y\in\mathcal{B}(\mathcal H)
\end{align*}
 for $\langle Y, X\rangle_{\rho_\beta}:=\textrm{Tr}[\rho_\beta^{1/2}Y^\dagger\rho_\beta^{1/2}X]$ the KMS inner product with respect to the (full rank) state $\rho_\beta$. This property implies that $\rho_\beta$ is a fixed point of the Lindblad evolution. Further details on KMS symmetric Lindbladians are provided in Appendix \ref{app:KMSsymm}.

The specific families of Lindbladians we consider arise as approximations to the unitary dynamics of a system weakly coupled to a thermal bath as in Figure \ref{fig:system-bath-couplings}, with a general Hamiltonian given by:
\begin{align}\label{eq:system-bath-H}
    H(t) = H_S + \alpha f(t)V + H_B
\end{align}
where $H_S$ is the system Hamiltonian, $H_B$ is the bath Hamiltonian and $V = \sum_aA^a\otimes B^a$ is the system-bath interaction. We also consider averages over different realizations of \eqref{eq:system-bath-H}. In what follows, we consider that the set $\{A^a\}$ corresponds to either single-qubit Paulis or Majorana operators. We are interested in the reduced dynamics on the subsystem $S$ after the joint system has evolved under the time-dependent unitary $U(t,t') =\mathcal T \exp\left(-i\int_{t'}^tH(s)ds\right) $ for some times $t'$ to $t$:
\begin{align*}
    \mathbb{E}\Big[\operatorname{Tr}_B\left[U(t,t')(\;\bigcdot\, \otimes \rho_\beta^B)U^\dagger(t,t')\right]\Big]
\end{align*}
where the initial state of the $B$ subsystem is the Gibbs state $\rho_\beta^B = \frac{e^{-\beta H_B}}{\operatorname{Tr}\left[e^{-\beta H_B} \right]}$ and the expectation over the different realizations of $H(t)$ has been taken.
\smallskip 
\paragraph{Repeated interaction models.} First, we consider algorithms based on repeated interactions with ancillary systems \cite{ding2025endtoendefficientquantumthermal,hahn2025provably,lloyd2025quantumthermalstatepreparation,hagan2025thermodynamiccostignorancethermal,Lambert_2024Fixing,lloyd2025quasiparticlecoolingalgorithmsquantum,Ramon_Escandell_2025} as applied to Gibbs state preparation. In particular, we consider a version of the algorithm in \cite{ding2025endtoendefficientquantumthermal}, which we refer to hereafter as \textit{repeated interaction Gibbs sampling}, where a single qubit bath Hamiltonian $H_B = -\frac{\omega}{2}Z$ with a randomly sampled frequency $\omega$ is coupled to a randomly sampled single site Pauli or Majorana operator $A^a$ (where $a$ indexes both the site and the type of Pauli or Majorana) from the set $\cA = \{A^a,-A^a\}_a$  via $ V = A^a\otimes X$ by the Gaussian 
\begin{align*}
    f(t)\to f_\kappa(t) = \frac{e^{-t^2/(\kappa\beta)^2}}{\sqrt{(\pi/2)^{1/2}\kappa\beta}}
\end{align*}
with width $\kappa\beta$. The algorithm consists of repeated applications of channel:

\begin{align}\label{eq:channel-near-term-gibbs}
    \Phi_{(\alpha,\kappa)} \!=\!\mathbb E_{A^a,\omega}\!\Big[\operatorname{Tr}_B\!\left[U(T,-T)(\;\bigcdot\, \otimes \rho_\beta^B)U^\dagger(T,-T)\right]\!\Big]
\end{align}
 where the corresponding Hamiltonian \eqref{eq:system-bath-H} is evolved from time $-T$ to $T$, and where the expectation is over the jump operators $A^a\in\cA$ and the bath frequency $\omega$. 
It was shown in \cite{ding2025endtoendefficientquantumthermal} that, for $\kappa = \cO(\text{poly}(N,\epsilon^{-1}))$ where $N$ is the system size, 
repeated applications of \eqref{eq:channel-near-term-gibbs} converges to a state $\epsilon$ close in $1$-norm to the Gibbs state $\rho_\beta^S = \frac{e^{-\beta H_S}}{\operatorname{Tr}\left[e^{-\beta H_S}\right]}$. Furthermore, for an appropriate choice of distribution of $\{\omega\}$, the dynamics of \eqref{eq:channel-near-term-gibbs} is well approximated by the Lindbladian evolution: 
\begin{align}
\Lm^{\operatorname{RI}}_\kappa \equiv  -i[H^{\operatorname{LS}},\bigcdot\;] + J\Lm^{\operatorname{G}}_{\kappa}.
\end{align}
Here, $\cL_{\kappa}^{\operatorname{G}}$ is an extensively normalized variant of the KMS symmetric Lindbladian introduced in \cite{chen2023efficient} with respect to $\rho_\beta^S$ and with jump operators:\begin{align}\label{eq:jump-operators-kappa}
    \widehat{A}_\kappa^a(\omega) =  \int_{-\infty}^\infty f_\kappa(t)A^a(t)e^{-i\omega t}dt,
 \end{align}
 where $A^a(t) := e^{iH_St}A^ae^{-iH_St}$. $H^{\operatorname{LS}}$ is the so-called \emph{Lamb shift} Hamiltonian, and $J$ is a normalization factor that here scales as $\Theta(\beta/N)$. 
 Both the Hamiltonian simulation cost and the spatial extent of the jump operators \eqref{eq:jump-operators-kappa} are proportional to $\kappa$. When $\kappa=\cO(1)$
 they are quasi-local and indeed $\cL_1^{\operatorname{G}}$ corresponds to the Lindbladians known to be fast mixing for the situations discussed in the introduction \cite{rouzé2024efficientthermalizationuniversalquantum,rouzé2024optimalquantumalgorithmgibbs,tong2025fastmixingweaklyinteracting,zhan2025rapidquantumgroundstate,smid2025rapidmixingquantumgibbs,bergamaschi2025quantumspinchainsthermalize}. Quasi-locality is an essential ingredient in the proofs of these results, which do not apply when $\kappa$ is not a constant.  We stress that even though the algorithm is modeled by Lindbladians with non-local jump operators, the implementation of the algorithm is local once a suitable Hamiltonian simulation subroutine is chosen (i.e. it can be compiled into a sequence of local gates with minimal overhead).

\smallskip

\paragraph{Macroscopic baths} A second model we consider is given by choosing the system-bath interaction  with $f(t)=1$, where the bath is understood as a single large uncontrolled system, such that $i)$ it is initially uncorrelated with the system, and $ii)$ it is a Gibbs state $\rho_\beta^B = \frac{e^{-\beta H_B}}{\operatorname{Tr}\left[e^{-\beta H_B} \right]}$ of a quadratic Hamiltonian.
Additionally, we also assume that the bath correlation functions $C_{aa}(t)=\tr{B^a(t)B^a \rho_\beta^B}$ have well-defined timescales $\Gamma^a_0= \int_{-\infty}^\infty \operatorname{d} t \vert C_{aa}(t) \vert $ and $\Gamma^a_0 \tau^a_0= \int_{-\infty}^\infty \operatorname{d} t \vert t C_{aa}(t) \vert $. These are standard in the theory of open systems \cite{stefanini2025lindbladme}. 

In \cite{scandi2025thermalizationopenmanybodysystems}, it was shown that the system's evolution $\rho_S(t) =\operatorname{Tr}_B[e^{-iHt }(\rho_{S}(0)\otimes \rho_\beta^B) e^{iHt }]$ is well approximated by the dynamics generated by a KMS Lindbladian $\mathcal{L}^{\operatorname{MB}}_\alpha$ satisfying (see Lemma \ref{le:approxKMS})
\begin{align} \label{eq:closenessKMSintro}
	&\| \rho_S(t)\!-\! e^{t \alpha^2 \mathcal{L}^{\operatorname{MB}}_\alpha}\rho_S(0)\|_1 
	\!=\!  \mathcal{O}\!\left(\alpha^3 \Gamma^2 t \!\left(  \!\tau \!+\!  \beta e^{(\alpha \Gamma \beta )^2}  \right) \right)\!,
\end{align}
where $\Gamma= \sum_a \Gamma_0^a$, $\Gamma \tau = \sum_a \Gamma_0^a \tau_0^a$, $\gamma_{\max}:=\max_a\Gamma^a_0$,
and where $\mathcal{L}^{\operatorname{MB}}_\alpha$ has jumps and rates related to $C_a(t)$ and its Fourier transform. In this case, one can avoid the need for a Lamb-shift term $H^{\operatorname{LS}}$ by renormalizing the system Hamiltonian. The Lindbladian explicitly depends on the coupling $\alpha$ (see Appendix \ref{app:physical}), and the jump operators become less local as $\alpha$ decreases and \eqref{eq:closenessKMSintro} becomes smaller. In particular, for $\alpha$ scaling inverse polynomially with the system size, a condition necessary for the approximation \eqref{eq:closenessKMSintro} to become precise, there is no known general tool to prove the fast convergence of the dynamics generated by $\mathcal{L}^{\operatorname{MB}}_\alpha$, to the best of our knowledge.


\section{Main Results}\label{sec:main-results}
\paragraph{Polynomial time Gibbs sampling via repeated interactions.} For our first result, we show that repeated interaction Gibbs sampling can prepare Gibbs states in polynomial time for a variety of systems consisting of lattice models of $N$ sites in the following regimes: high temperature local qubit models \cite{rouzé2024efficientthermalizationuniversalquantum}, weakly interacting fermionic models at any constant temperature \cite{tong2025fastmixingweaklyinteracting}, and {$1$-D} spin chains at any constant temperature \cite{bergamaschi2025quantumspinchainsthermalize}. We present the cost of the algorithm in terms of the total Hamiltonian simulation time required to prepare a state $\epsilon$-close to $\rho_\beta^S$ in $1$-norm. This is given by $t_{\operatorname{total}}(\epsilon)= 2T \times t_{\textup{therm}}(\epsilon)$ where $T$ is the time parameter appearing in \eqref{eq:channel-near-term-gibbs} and $t_{\textup{therm}}(\epsilon)$ is the number of times \eqref{eq:channel-near-term-gibbs} needs to be applied to ensure that $\|\Phi_{(\alpha,\kappa)}^{t_{\textup{therm}}(\epsilon)}(\sigma)-\rho_\beta^S\|_1\leq \epsilon$ for all $\sigma$. In what follows, we refer to $t_{\textup{therm}}(\epsilon)$ as the \textit{thermalization index}. We remain agnostic as to which Hamiltonian simulation subroutine is used, noting that once one is chosen the elementary gate cost of the algorithm will typically be increased by a factor of $\cO(\|H_S\|)$ \cite{Haah_2021}. We combine all of these results in the following corollary:
\begin{thm}\label{thm:repeated-interaction-efficient-prep}
Repeated interaction Gibbs sampling prepares a state $\epsilon$ close in $1$-norm to the Gibbs state $\rho_\beta^S$ with total Hamiltonian simulation time: 
$$t_{\textup{total}}(\epsilon)=\widetilde{\cO}\left( \frac{N^{10}}{\epsilon^4}\right)$$ for sufficiently small $\beta \le \beta^*= \mathcal{O}(1)$ for  $(k,l)$-local Hamiltonians with maximum interaction strength $h$ as in \cite{rouzé2024efficientthermalizationuniversalquantum}, and at any constant $\beta$ for local weakly interacting Fermionic Hamiltonians with sufficiently small maximum interaction strength $U \le U^*=\mathcal{O}(1)$ as in \cite{tong2025fastmixingweaklyinteracting}, as well as $1$-D nearest neighbor Hamiltonians with local dimension $2^q$ as in \cite{bergamaschi2025quantumspinchainsthermalize}. 
\end{thm}
\noindent Here, the notation $\widetilde{\cO}$ suppresses subleading polylogarithmic factors in $N$ and $\epsilon$.
\noindent We refer to Section~\ref{sec:methods} and Appendix \ref{app:RIM} for proof details. Similar results likely apply to the weakly interacting qubit cases considered in \cite{smid2025rapidmixingquantumgibbs}, where the authors show a lower bound on the spectral gap of the Lindbladian from \cite{ding2024efficientquantumgibbssamplers}.

We expect that this result can be improved. A factor of $N^6/\epsilon^3$ comes from the effective slowing of the dynamics due to $\alpha$ needing to be small in order to control an error term that comes from truncating the Dyson expansion of $U(t,t') =\mathcal T \exp\left(-i\int_{t'}^tH(s)ds\right)$. In recent work, \cite{wang2025lindbladdynamicsrigorousguarantees} the corresponding error term has a considerably better error scaling with $\alpha$. It would be interesting to see if their results can be applied to our setting. 

\smallskip

\paragraph{Efficient physical thermalization} We also prove fast convergence of the generators originated through weak coupling with a bath \cite{scandi2025thermalizationopenmanybodysystems} at high temperatures.
Given \eqref{eq:closenessKMSintro}, the generators only resemble the exact system's dynamics for very small $\alpha$, where they are too non-local for mixing time proofs to apply. Our methods below allow us to overcome this.

In order for the bath to thermalize well, we need two physically motivated assumptions on its correlation functions, summarized as (see Appendix \ref{app:physical} for the precise formulation)
\begin{itemize}
    \item The bath correlation function decays rapidly with time.
    \item The timescale is $\tau = \mathcal{O}(\beta)$ at high temperatures.
\end{itemize}
These essentially imply that the information dissipates sufficiently quickly in the bath and that it induces white noise at high temperatures, so that in particular at $\beta=0$ one obtains depolarizing noise.
Under these, we obtain the following result.

\begin{thm}\label{thm:bath}
    Let $H$ be a $(k,l)$ local Hamiltonian, and the bath be at a temperature $\beta \le \beta^*$ for some $\beta^*= \mathcal{O}(1)$. Then, at a time $t = \widetilde{\mathcal{O}}\left( N^7 \epsilon^{-2} \right)$ and a sufficiently weak coupling $\alpha= \widetilde{\mathcal{O}}\left( N^{-3} \epsilon \right )$, the system's state $\rho_S(t)$ is $\epsilon$-close to the Gibbs state in 1-norm.
\end{thm}

\noindent In Appendix \ref{app:physical}, we show that Theorem \ref{thm:bath} also straightforwardly extends to the related setting of a product (or 1-local) Hamiltonian $H_0=\sum_i h_i$ to which we add a quasilocal perturbation $H=H_0+\lambda V$ of bounded strength $\lambda \le \lambda^*$, a setting previously studied for the generators of \cite{chen2023efficient,ding2024efficientquantumgibbssamplers} in \cite{smid2025rapidmixingquantumgibbs}.

In Theorem \ref{thm:bath} we required the coupling constant to be inversely decaying with system size, which also means the time grows as a polynomial in $N$. This is because we are using a rather strong notion of closeness to the Gibbs state, where the final state is close in a (global) 1-norm. If one only requires closeness in local observables, Lieb-Robinson bounds would guarantee that a coupling constant independent of system size suffices.

\smallskip

\paragraph{Faster thermalization for commuting models}
We also consider the generators from macroscopic baths \cite{scandi2025thermalizationopenmanybodysystems} when the system's Hamiltonian is commuting. Commuting models have highly degenerate frequencies, so that the generator quickly approaches that of Davies' as $\alpha$ becomes small (see Appendix \ref{app:davies}). For these, a recent body of work provides tighter bounds on thermalization times \cite{capel2020modified,bardet2023rapid,kochanowski2024rapidthermalizationdissipativemanybody,stengele2025modified}. In particular, it was recently shown that for CSS codes in low dimensions akin to the $2$ dimensional toric code, their associated Davies dynamics satisfies a faster exponential decay rate known as the modified logarithmic Sobolev inequality \cite{stengele2025modified}. Combined with an  analogous bound to \eqref{eq:closenessKMSintro} but for approximation of the dynamics by Davies, we derive the following result:

\begin{thm}\label{thmCSSmemory}
The Macroscopic Bath dynamics associated to the $2$-dimensional Toric code with coupling constant $\alpha=\widetilde{\mathcal{O}}\left(\frac{\epsilon}{\gamma_{\min}^2N^2 \tau}\right)$ thermalizes to its Gibbs state in time $t=\widetilde{\mathcal{O}}\Big(\frac{\gamma_{\max}^4N^4\tau^2}{\epsilon^2}\Big)$. 
\end{thm}

\noindent The same conclusion extends to a broader class of CSS codes \cite{stengele2025modified}. From a physical perspective, this result provides the first complete dynamical justification of the long-standing claim that CSS codes in low lattice dimensions cannot serve as good quantum memories \cite{yoshida2011feasibility}. A formal proof of Theorem~\ref{thmCSSmemory} is given in Appendix~\ref{app:davies}.


\section{Methods}\label{sec:methods} 
\noindent Our results rely on the following lemma, which shows that the spectral gaps of certain single parameter families of Lindbladians are monotonic. 
For the Lindbladian $\mathcal{L}_\kappa^{\operatorname{G}}$, the parameter corresponds to the width $\kappa$ of the function $f_\kappa$ \eqref{eq:jump-operators-kappa}, while for the Lindbladian in \cite{scandi2025thermalizationopenmanybodysystems} it is the coupling strength $\alpha$. Both of these quantities determine the locality and the precision of the jump operators, with both increasing monotonically as $\kappa$ increases and $\alpha$ decreases respectively.
We recall that the spectral gap $\lambda_{\operatorname{Gap}}(\cL)$ of a KMS symmetric Lindbladian is defined as its smallest non-zero eigenvalue in absolute value. It controls the speed at which the evolution generated by $\cL$ converges to its Gibbs state:
\begin{align*}
\big\|e^{t\mathcal{L}}(\sigma)-\rho_{\beta}\big\|_1\le2 e^{-\lambda_{\operatorname{Gap}}(\cL) t}\|\rho_\beta^{-1}\|. 
\end{align*}

\begin{lem}\label{lem:CKG-monotonicity}
The following gap monotonicity statements hold true for the families $\{\cL^{\operatorname{G}}_\kappa\}_{\kappa\ge 1}$ and $\{\cL^{\operatorname{MB}}_\alpha\}_{\alpha\ge 0}$:
\begin{enumerate}
    \item\label{item:Gaussian-Paramerization} For any two $1\!\le\! \kappa\!\le \!\kappa'\!<\!\infty$, $\lambda_{\operatorname{Gap}}(\cL^{\operatorname{G}}_\kappa)\le \lambda_{\operatorname{Gap}}(\cL^{\operatorname{G}}_{\kappa'})$; 
       
    \item\label{item:Open-System-Paramerization} For any two $0<\alpha\le\alpha'$, $\lambda_{\operatorname{Gap}}(\cL^{\operatorname{MB}}_{\alpha'})\le \lambda_{\operatorname{Gap}}(\cL^{\operatorname{MB}}_{\alpha})$.
\end{enumerate}
\end{lem}
\noindent Lemma \ref{lem:CKG-monotonicity} thus allows us to connect the gaps of two complementary regimes: that in which the dynamics are accurately approximated by Lindbladians, and that in which spectral gap estimates are within reach. 
It is proven in Appendix \ref{app:KMSsymm} and relies on the ability to write the Dirichlet forms of the associated Lindbladians as convolutions. An identical result holds for the closely related \emph{modified logarithmic Sobolev inequality} (MLSI) constant, which controls the rate at which the entropy of the evolved state converges to that of the Gibbs state $\rho_\beta^S$
\cite{Kastoryano_2013}; see Appendix \ref{app:KMSsymm} for more details. 

\smallskip

\paragraph{Repeated interaction Gibbs sampling.} We now show how Lemma~\ref{lem:CKG-monotonicity} can be used to prove the efficient preparation results presented in Section~\ref{sec:main-results} for Gibbs sampling based on repeated interactions. In order to ensure that the fixed point of the channel \eqref{eq:channel-near-term-gibbs} is $\epsilon$ close to the Gibbs state $\rho_\beta^S$ the parameter $\kappa$ typically needs to scale as a $\text{poly}(N,\epsilon^{-1})$, which in turn leads to jump operators \eqref{eq:jump-operators-kappa} with extensive support. Lemma~\ref{lem:CKG-monotonicity} allows us to bootstrap existing spectral gap lower bound proofs that require quasi-local Lindbladians to the non-local ones considered here. We formalize this in the following lemma, which provides an upper bound on the total Hamiltonian simulation time $t_{\text{total}}(\epsilon)$ needed to prepare a state $\epsilon$-close to $\rho_\beta^S$ whenever a lower bound on the gap for the Lindbladian $\Lm^{\operatorname{G}}_1$ can be shown. 
\begin{lem}\label{lem:gap-implies-efficient-prep-epsilon}
    Let $\left\{\Lm^{\operatorname{G}}_\kappa\right\}_{\kappa\ge 1}$ be a family of parametrized Lindbladians as defined in item \ref{item:Gaussian-Paramerization} of Lemma~\ref{lem:CKG-monotonicity} with fixed point $\rho_\beta^S$, and let $\lambda$ be a lower bound on the spectral gap of $\Lm^{\operatorname{G}}_1$. For the set $\cA$ from which the coupling operators $A^a$ are sampled, define $\|\cA\|_{\operatorname{G}}= \|\sum_aA ^{a\dagger}A^a\|$. Then for any $\epsilon>0$ the repeated interaction Gibbs sampling algorithm prepares a state $\sigma$ such that $\|\sigma-\rho^S_\beta\|_1\leq \epsilon$ using total Hamiltonian simulation time:
    \begin{align*}
        t_{\textup{total}}(\epsilon) = \widetilde{\cO}\left( \frac{\beta\|\cA\|_{\operatorname{G}}^5}{\epsilon^4\lambda^5}\log\left( \|\rho_\beta^{S\;-1}\|\right)^5\right),
    \end{align*}
     where the $\widetilde{\cO}$ notation suppresses subleading polylogarithmic factors in $\beta, \lambda, \epsilon$,$\|\cA\|_{\operatorname{G}}$, and $\log\left( \|\rho_\beta^{S\;-1}\|\right)$.
\end{lem}
\noindent The proof, provided in Appendix \ref{app:RIM}, requires handling the Lamb Shift term $-i[H^{\operatorname{LS}},\bigcdot\; ]$ as well as ensuring that the parameters $\alpha,\kappa$ and $T$ can be picked such that the spectral gap $\lambda_{\operatorname{Gap}}(\cL^{\operatorname{G}}_\kappa)$ of the dynamics generated by $\Lm^{\operatorname{G}}_\kappa$ and the thermalization index $t_{\textup{therm}}(\epsilon)$ of $\Phi_{(\alpha,\kappa)}$ are relatable, while the fixed point error of $\Phi_{(\alpha,\kappa)}$ is simultaneously controllable. Since the Hamiltonian \eqref{eq:system-bath-H} used in this algorithm satisfies $\max_t\|H(t)\|  = \Theta(\|H_S\|)$, its elementary gate cost will be increased by a factor of $\cO(\|H_S\|)$ when a concrete Hamiltonian simulation subroutine is chosen, potentially along with extra factors of $\epsilon^{-1}$ depending on the choice of subroutine. We again note that the high powers of $\lambda$, $\epsilon$, $\|\cA\|_{\operatorname{G}}$, and $\log( \|\rho_\beta^{S\;-1}\|)$ are likely not optimal and may be improved with the results from \cite{wang2025lindbladdynamicsrigorousguarantees}.

The results of Theorem~\ref{thm:repeated-interaction-efficient-prep} then follow simply from Lemma~\ref{lem:gap-implies-efficient-prep-epsilon} together with the existence of a system size independent lower bound on the spectral gap of the relevant $\cL^{\operatorname{G}}_1$ Lindbladians, and by noting that both $\log (\|\rho_{\beta}^{S\;-1}\|)$ and $\|\cA\|_{\operatorname{G}}$ are $\cO(N)$ for the considered models \cite{rouzé2024efficientthermalizationuniversalquantum,tong2025fastmixingweaklyinteracting}, and \cite{bergamaschi2025quantumspinchainsthermalize}. Formal statements and more details are provided in Appendix \ref{app:RIM}.

\smallskip

\paragraph{Efficient physical thermalization.}
Through a similar argument, Lemma \ref{lem:CKG-monotonicity} also allows us to prove Theorem \ref{thm:bath}. 
 Given \eqref{eq:closenessKMSintro}, the system's evolution $\rho_S(t)$ is sufficiently close to that of a KMS-Lindbladian evolution only when $\alpha$ is very small, when the jump operators have a very large support.  This is an issue for existing spectral gap estimates, since we can only prove them for quasi-local Lindbladians. Here this means that $\mathcal{L}^{\operatorname{MB}}_\alpha $ must have a larger $\alpha$, in tension with the bound in \eqref{eq:closenessKMSintro}. Using Lemma \ref{lem:CKG-monotonicity}, we instead obtain. 
 \begin{lemma}\label{lem:bathgap}
      Let $\left\{\Lm^{\operatorname{MB}}_\alpha\right\}_{\alpha \ge 0}$ be a family of parametrized Lindbladians as defined in item \ref{item:Open-System-Paramerization}  of Lemma~\ref{lem:CKG-monotonicity} with fixed point $\rho_\beta^S$, and let $\lambda$ be a lower bound on the spectral gap of $\Lm^{\operatorname{MB}}_{\alpha^*}$ with $\alpha^*= \frac{\sqrt{2+3 \Gamma \tau}}{2 \beta \Gamma}$. Then for any $\epsilon>0$ the system-bath interaction prepares a state $\sigma$ such that $\|\sigma-\rho^S_\beta\|_1\leq \epsilon$ after time
      \begin{equation}
         \widetilde{ \mathcal{O}} \left( \frac{\Gamma^4 (\tau+\beta)^2\log\left( \|\rho_\beta^{S\;-1}\|\right)^3 }{\lambda^3 \epsilon^3} \right)
      \end{equation}
 \end{lemma}
 
The connection of the spectral gaps at large and small $\alpha$ is then done through the use of Lemma \ref{lem:CKG-monotonicity}. From this, the proof of \ref{thm:bath} requires lower bounds on the gaps. This is shown in Appendix \ref{app:physical}, where it is done from first principles following that of \cite{rouzé2024optimalquantumalgorithmgibbs}, and aided by the assumptions on the bath correlation functions. The dependence of $\alpha$ with $N$ is determined by the $\Gamma$ dependence in \ref{eq:closenessKMSintro}, which can likely be improved. 

\smallskip

\section{Discussion and Conclusion}\label{sec:discussion}

\noindent We have demonstrated the efficient preparation of Gibbs states for both near term quantum Gibbs sampling algorithms and thermalization models for many-body quantum systems. In doing so we introduced a technical result that greatly extends the regimes in which fast mixing results can be proven for KMS symmetric Lindbladians. Quasi-locality of the jump operators was a necessity of previously existing proofs, whereas now we are able to extend them to manifestly non-local Lindbladians. 

We believe it is possible to extend Theorem \ref{thm:bath} to weakly interacting fermions and 1D to the generator, which require extending the mixing time bounds from \cite{smid2025rapidmixingquantumgibbs,bergamaschi2025quantumspinchainsthermalize} to our setting. Another important question is what is the most general set of conditions on the bath that guarantees fast convergence. 

While our results allow the extension of MLSI constants for quasi-local Lindbladians to non-local ones, there are currently no examples of MLSI constants for non-commuting systems. Instead, rapid mixing (converging in time $\cO(\log(N))$) results in this regime have been proven via an approach using the so-called oscillator norm \cite{rouzé2024optimalquantumalgorithmgibbs,smid2025rapidmixingquantumgibbs,zhan2025rapidquantumgroundstate}. Interesting open problems are thus proving MLSI constants for quasi-local Lindbladians and seeing if a similar monotonicity result can be derived for the oscillator norm method.

\emph{Acknowledgements—}

CR would like to thank Anthony Chen for discussions that led to the establishment of Lemma 1.
CR and DSF are supported by France
2030 under the French National Research Agency award number “ANR-22- PNCQ-0002”. DSF acknowledges funding from the EERC grant GIFNEQ 101163938 and the project EQUALITY and from the Novo Nordisk
Foundation (Grant No. NNF20OC0059939 Quantum for Life). AMA acknowledges support from the Spanish Agencia Estatal de Investigacion through the grants ``IFT Centro de Excelencia Severo Ochoa CEX2020-001007-S" and ``Ram\'on y Cajal RyC2021-031610-I'', financed by MCIN/AEI/10.13039/501100011033 and the European Union NextGenerationEU/PRTR. This
project was funded within the QuantERA II Programme that has received funding from the EU’s H2020 research and innovation programme under the GA No 101017733.

\bibliographystyle{apsrev4-2}
\bibliography{references}

\appendix
\pagestyle{myheadings}
\onecolumngrid

\section{KMS Symmetric Lindbladians and mixing times}\label{app:KMSsymm}
\noindent In this section we review the relevant properties of KMS symmetric Lindbladians and prove Lemma~\eqref{lem:CKG-monotonicity}. 
The Kubo-Martin-Schwinger (KMS) inner product with respect to the full rank state $\rho_\beta$ on the Hilbert space $\cH$ is given by:
\begin{align*}
    \langle X,Y \rangle_{\rho_\beta}:= \operatorname{Tr}[\rho_\beta^{1/2}X^\dagger\rho_\beta^{1/2}Y].
\end{align*}
The KMS inner product can be used to define a $\rho_\beta$ weighted version of the $2$-norm via $\|X\|_{\rho_\beta}^2:= \langle X,X \rangle_{\rho_\beta}$.
A Heisenberg picture Lindbladian $\Lm^\dagger$ is said to be KMS symmetric if it is self adjoint with respect to the KMS inner product (this property is also often called $\rho_\beta$ detailed balance). That is:
\begin{align*}
    \langle X,\Lm^\dagger(Y) \rangle_{\rho_\beta} = \langle \Lm^\dagger(X),Y \rangle_{\rho_\beta} \;\forall \;X,Y.
\end{align*}
This property ensures that $\Lm(\rho_\beta) =0$ and thus that the evolution generated by $\Lm$ converges to $\rho_\beta$.
Throughout the following we will assume that $\rho_\beta = \frac{e^{-\beta H}}{\operatorname{Tr}[e^{-\beta H}]}$ for a Hamiltonian $H$ with the eigendecomposition $H = \sum_{i}E_i \Pi_i$ over the Hilbert space $\cH$ for a complete set of orthogonal projectors $\Pi_i$. The Bohr frequencies of $H$ are defined as the set:
\begin{align*}
    B_H := \{\nu = E_i-E_j: E_i,E_j \in \operatorname{spec}(H)  \}.
\end{align*}
Given an operator $A$ and a Bohr frequency $\nu\in B_H$ , we define the operator $A_\nu$ as $A$ where its allowed transitions are restricted to energy subspaces separated by $\nu$:
\begin{align*}
    A_\nu := \sum_{i,j:E_i-E_j = \nu} \Pi_i\, A\, \Pi_j.
\end{align*}
A general form for KMS symmetric Lindbladians in the Heisenberg picture is then given by: 
\begin{align}\label{eq:KMSLind}
\mathcal{L}^\dagger(X)=i\sum_{\nu\in B_H}[B_\nu,X]+\sum_{a}\sum_{\nu_1,\nu_2\in B_H}\Lambda^a_{\nu_1,\nu_2}\left(A^{a\dagger }_{\nu_2} X A^a_{\nu_1}-\frac{1}{2}\{A^{a\dagger}_{\nu_2} A^a_{\nu_1},X\}\right)
\end{align}
for some coefficients $\Lambda^a_{\nu_1,\nu_2}$ satisfying $\Lambda^a_{\nu_1,\nu_2}=\overline{\Lambda^a}_{-\nu_1,-\nu_2}e^{-\beta \left(\frac{\nu_1+\nu_2}{2} \right)}$, and 
\begin{align}\label{eq:KMS-lindblad-B}
B_\nu:=\sum_{a}\sum_{\nu_1-\nu_2=\nu}\frac{\operatorname{tanh}(-\beta(\nu_1-\nu_2)/4)}{2i}\Lambda^a_{\nu_1,\nu_2}A_{\nu_2}^{a\dagger}A^a_{\nu_1}\,.
\end{align}
An important object in the study of Lindbladians is the Dirichlet form, defined for a Lindbladian $\Lm$ via:
\begin{align*}
    \mathcal E(X,Y) := -\langle X,\Lm^\dagger(Y)\rangle_{\rho_\beta}.
\end{align*}
 Given a KMS symmetric Lindbladian $\Lm$ in the general form \eqref{eq:KMSLind}, its Dirichlet form is given by (see \cite{rouzé2024efficientthermalizationuniversalquantum}):
\begin{align*}
    \mathcal E(X,Y) = \sum_{a,\nu_1,\nu_2}\frac{\Lambda^a_{\nu_1,\nu_2}e^{\beta(\nu_1+\nu_2)/4}}{\cosh(\beta(\nu_1-\nu_2)/4)}\langle[A^a_{\nu_1},X],[A^a_{\nu_2},Y] \rangle_{\rho_\beta}.
\end{align*}
\paragraph{Spectral gaps and fast mixing}
The Dirichlet form can be used to define a variational form of the spectral gap of a Lindbladian via the optimization:
\begin{align}\label{gap:varexpre}
\lambda_{\operatorname{Gap}}(\cL) := \inf_{X\in\cB_{\operatorname{sa}}(\cH)}\frac{\E(X)}{\|X-\langle\mathbb I,X \rangle_{\rho_\beta}\mathbb I\|_{\rho_\beta}^2},   
\end{align}
where $\cB_{\operatorname{sa}}(\cH)$ stands for the subspace of self-adjoint matrices over $\cH$, and we denote $\E(X)\equiv \E(X,X)$. The spectral gap is well-known to control the mixing time of the evolution generated by $\cL$: 
\begin{align*}
t_{\operatorname{mix}}(\epsilon):=\inf\left\{ t: \left\|e^{t\Lm}(\sigma)-\rho_\beta\right\|_1 \leq \epsilon \; \forall \; \sigma \in \mathcal S(\mathcal H) \right\},
\end{align*}
where $\|\bigcdot\|_1$ stands for the trace distance on the algebra $\mathcal{B}(\cH)$ of matrices over $\cH$, and $\mathcal{S}(\mathcal{H})$ stands for the set of quantum states. Indeed,
\begin{align}\label{eq:gap-mixing-time-ineq}
\big\|e^{t\mathcal{L}}(\sigma)-\rho_{\beta}\big\|_1\le2 e^{-\lambda_{\operatorname{Gap}}(\cL) t}\|\rho_\beta^{-1}\|\qquad \Rightarrow \qquad t_{\operatorname{mix}}(\epsilon) \le \frac{1}{\lambda_{\operatorname{Gap}}}\log\left(\frac{2\|\rho_\beta^{-1}\|}{\epsilon}\right). 
\end{align}
\paragraph{Modified logarithmic Sobolev inequality and rapid mixing}
A much tighter bound on the mixing time can be obtained using entropy decay estimates: we recall Umegaki's definition of the quantum relative entropy between any state $\sigma\in\mathcal{S}(\cH)$ and $\rho_\beta$:
\begin{align*}
D(\sigma\|\rho_\beta):=\tr{\sigma(\log(\sigma)-\log(\rho_\beta))}.
\end{align*}
The semigroup $e^{t\cL}$ satisfies the \textit{exponential entropy decay} with rate $\alpha>0$ if for any $\sigma\in \mathcal{S}(\cH)$ and $t\ge 0$,
\begin{align*}
D(e^{t\cL}(\sigma)\|\rho_\beta)\le e^{-\alpha t}D(\sigma\|\rho_\beta).
\end{align*}
Via Pinsker's inequality, the entropy decay bound leads to the following bound on the mixing time:

\begin{align}\label{eq.MLSIMixing}
\|e^{t\cL}(\sigma)-\rho_\beta\|_1\le \sqrt{2e^{-\alpha t}D(\sigma\|\rho_\beta)}\le \sqrt{2 e^{-\alpha t}\log(\|\rho_\beta^{-1}\|)}\qquad \Rightarrow \qquad t_{\operatorname{mix}}(\epsilon)\le \frac{1}{\alpha}\log\left(\frac{2\log(\|\rho_\beta^{-1}\|)}{\epsilon^2}\right).
\end{align}

\noindent Next, we define the \textit{entropy production} of $\rho$ as \cite{Spohn1978}:
\begin{align*}
\operatorname{EP}(\sigma):=-\frac{d}{dt}D(e^{t\cL}(\sigma)\|\rho_\beta)=-\tr{\mathcal{L}(\sigma)(\log\sigma-\log\rho_\beta)}.
\end{align*}
By Gr\"{o}nwall's inequality, the semigroup generated by $\cL$ satisfies the exponential entropy decay with rate $\alpha>0$ if and only if, for any state $\sigma\in \mathcal{S}(\cH)$,
\begin{align}\label{eq:MLSI}
\alpha\,D(\sigma\|\rho_\beta)\le \operatorname{EP}(\sigma).
\end{align}
The inequality \eqref{eq:MLSI} is known as the $\alpha$-\textit{modified logarithmic Sobolev inequality} (MLSI).
The best constant $\alpha$ satisfying such bounds is called the MLSI constant and satisfies the optimization
\begin{align*}
\alpha_{\operatorname{MLSI}}(\cL):=\inf_{\sigma\in\mathcal{S}(\cH)}\frac{\operatorname{EP}(\sigma)}{D(\sigma\|\rho_\beta)}.
\end{align*}
\paragraph{A monotonicity lemma for the spectral gaps and MLSI constants of KMS symmetric Lindbladians.}  We now prove the folowing Lemma, which shows that certain $1$-parameter families of KMS symmetric Lindbladians have monotonically increasing spectral gaps and MLSI constants when the parameterization satisfies a certain semigroup property with respect to convolution. This will form the basis for the eventual proof of Lemma~ \ref{lem:CKG-monotonicity}.
 
\begin{lemma}\label{lem:MLSI-gap}
   Let $\{\mathcal{L}_\delta\}_{\delta\ge 0}$ be a parametrized family of primitive, KMS symmetric Lindbladians with same fixed point $\rho_\beta$. Assume that for all $X,Y\in\mathcal{B}_{\operatorname{sa}}(\mathcal{H})$, given the function $\mathcal{E}^{(\delta)}_{XY}(s):=-\langle Y(s\beta),\mathcal{L}_\delta^\dagger(X(s\beta))\rangle_{\rho_\beta}$, where $X(s):=e^{isH}Xe^{-isH}$, takes the form
   \begin{align}\label{eq:convolution}
\mathcal{E}^{(\delta)}_{XY}= f_\delta\star  \mathcal{E}_{XY}^{(\infty)}
   \end{align}
  for an integrable function $f_\delta:\mathbb{R}\to\mathbb{R}_+$ which satisfies the following semigroup property: for any $\delta\le\delta'$,
  \begin{align}\label{eq:semigroupprop}
f_\delta =g_{\delta,\delta'}\star f_{\delta'}
  \end{align}
  for some function $g_{\delta,\delta'}:\mathbb{R}\to\mathbb{R}_+$. Then, for all $\delta\le \delta'$,
   \begin{align}\label{eq:gap}
  \lambda_{\operatorname{Gap}}(\cL_\delta)\ge \|g_{\delta,\delta'}\|_1\,\lambda_{\operatorname{Gap}}(\cL_{\delta'})\qquad \text{ and }\qquad \alpha_{\operatorname{MLSI}}(\cL_\delta)\ge \|g_{\delta,\delta'}\|_1\,\alpha_{\operatorname{MLSI}}(\cL_{\delta'}).
   \end{align}
\end{lemma}
\begin{proof}

The variational expression of the gap \eqref{gap:varexpre} associated to parameter choice $\delta$ is:
\begin{align*}
    \lambda_\delta = \inf_{X\in\mathcal{B}_{\operatorname{sa}}(\mathcal{H})} \frac{\Em_\delta(X)}{\|X-\langle\mathbbm I,X \rangle_{\rho_\beta}\mathbbm I\|_{\rho_\beta}^2}.
\end{align*}
By the condition \eqref{eq:convolution}, we get
\begin{align*}
\mathcal{E}_\delta(X)=-\langle X,\mathcal{L}^\dagger_{\delta}(X)\rangle_{\rho_\beta}=f_\delta\star  \mathcal{E}^{(\infty)}_{XX}(0).
\end{align*}
Using the semigroup property \eqref{eq:semigroupprop}, we can further write 
\begin{align*}
\mathcal{E}_\delta(X)=\big(g_{\delta,\delta'}\star f_{\delta'}\star \mathcal{E}^{(\infty)}_{XX}\big)(0)=\big(g_{\delta,\delta'}\star \mathcal{E}^{(\delta')}_{XX}\big)(0).
\end{align*}
Hence,
\begin{align*}
\frac{\mathcal{E}_\delta(X)}{\|X-\langle \mathbb{I},X\rangle_{\rho_\beta}\mathbb{I}\|^2_{\rho_\beta}}&=\int g_{\delta,\delta'}(s)\,\frac{\mathcal{E}_{\delta'}(X(-\beta s))}{\|X-\langle \mathbb{I},X\rangle_{\rho_\beta}\mathbb{I}\|_{\rho_\beta}^2}\,ds\\
&=\int g_{\delta,\delta'}(s)\,\frac{\mathcal{E}_{\delta'}(X(-\beta s))}{\|X(-\beta s)-\langle \mathbb{I},X(-\beta s)\rangle_{\rho_\beta}\mathbb{I}\|_{\rho_\beta}^2}\,ds.
\end{align*}
The result for the gap follows directly after lower bounding the ratio in the integrand by $\lambda_{\delta'}$ and minimizing over $X$ on the left-hand side. The result on the modified logarithmic Sobolev constant follows the same strategy: since the relative entropy is unitary invariant, $D(\sigma\|\rho_\beta)=D(\sigma(\beta s)\|\rho_\beta)$ for all $s$. Moreover, the entropy production associated to $\cL_\delta$ satisfies
\begin{align*}
\operatorname{EP}_\delta(\sigma(\beta s))=-\operatorname{Tr}\big(\mathcal{L}_\delta(\sigma(\beta s))(\log\sigma(\beta s)-\log\rho_\beta)\big)=-\langle \mathcal{L}_\delta^\dagger(X(s)),Y(s)\rangle_{\rho_\beta}\equiv \mathcal{E}^{(\delta)}_{XY}(s),
\end{align*}
where $X:=\rho_\beta^{-1/2}\sigma \rho_\beta^{-1/2}$ and $Y=\log(\sigma)-\log(\rho_\beta)$.

\end{proof}

\paragraph{Description of Lindbladians and proof of Lemma~\ref{lem:CKG-monotonicity}.}  We now describe the two families of Lindbladians to which Lemma~\ref{lem:CKG-monotonicity} applies. The first is a version \footnote{Our parametrization $\kappa$ is related to the parameters $\sigma_E$ and $\omega_\gamma$ in \cite{chen2023efficient} via $\sigma_E = \frac{1}{(\kappa\beta)}$ and $\omega_\gamma = \frac{1}{\beta}$ in \cite{chen2023efficient}. The version here also has a norm that scales linearly with the number of jump operators.} of the Lindbladian
from \cite{chen2023efficient}:
they are defined with respect to a filter function and its Fourier transform given by:
\begin{align*}
    f_\kappa(t) = \frac{e^{-\frac{t^2}{(\kappa\beta)^2}}}{(\pi/2)^{1/4}(\kappa\beta)^{1/2}} \qquad\operatorname{and}\qquad \widehat{f}_\kappa(\omega)=\int_{\mathbb R} f_\kappa(t)e^{-i\omega t}dt.  
\end{align*}
 Given a set of jump operators $A^a$, the Lindbladians in the Schrodinger picture are:
\begin{align}\label{eq:CKG-lindbladian}
    \Lm_\kappa^{\operatorname{G}}  = - i [B^{\operatorname{G}}_\kappa,\bigcdot\;] + \sum_a\int_{\mathbb R}\gamma_\kappa(\omega)\left( \widehat{A}_\kappa^a(\omega)(\,\bigcdot\,)\widehat{A}_\kappa^a(\omega)^\dagger -\frac{1}{2} \left\{\widehat{A}_\kappa^a(\omega)^\dagger \widehat{A}_\kappa^a(\omega),\bigcdot\; \right\}\right)d\omega
\end{align}
with 
\begin{align}\label{eq:jump-operators-kappa-app}
    \widehat{A}^a_\kappa(\omega) = \int_{\mathbb R} f_\kappa(t)A^a(t)e^{-i\omega t}dt = \sum_{\nu\in B_H}\widehat{f}_\kappa(\omega-\nu)A^a_{\nu}
\end{align}
for $A^a(t) = e^{iHt}A^ae^{-iHt}$ the Heisenberg evolution of $A^a$.
The function $\gamma_\kappa(\omega)$ is given by
\begin{align}\label{eq:KMS-gamma-app}
    \gamma_\kappa(\omega) = \exp\left( -\frac{(\beta\omega + 1)^2}{2\left(2-\frac{1}{\kappa^2}\right)}\right), \qquad \text{and satisfies} \qquad \gamma_\kappa(-\omega) = \exp\left(\frac{2\beta\omega}{\left(2-\frac{1}{\kappa^2}\right)}\right)\gamma_\kappa(\omega).
\end{align}
In terms of \eqref{eq:KMSLind}, the coefficients $\Lambda_{\nu_1,\nu_2}$ are given by:
\begin{align*}
    \Lambda_{\nu_1,\nu_2}^{\operatorname{G},\kappa} = \int_{\mathbb R} \gamma_\kappa(\omega)\widehat{f}_\kappa(\omega-\nu_1)\widehat{f}_\kappa(\omega-\nu_2)d\omega = 2\pi\sqrt{\frac{2-1/\kappa^2}{2}}\exp\left(-\frac{\kappa^2(\beta\nu_1 - \beta\nu_2)^2}{8}\right)\exp\left(-\frac{(\beta\nu_1 +\beta\nu_2+2)^2}{16} \right).
\end{align*}
The operator $B^{\operatorname{G}}_\kappa = \sum_\nu (B^{\operatorname{G}}_\kappa)_\nu$ is then given in terms of Bohr frequencies as in \eqref{eq:KMS-lindblad-B} and has a representation in the time domain given by:
\begin{align*}
    \sum_aB_\kappa^{\operatorname{G}} =\frac{1}{2}\sqrt{\frac{2-1/\kappa^2}{2}}\int_{\mathbb R}\int_{\mathbb R}b_1(s)b_2(r) A^a(\beta(s-r))A^a(\beta(s+r))dsdr
\end{align*}
for the functions $ b_1(s) =\frac{2}{\kappa\sqrt{\frac{\pi}{2}}} \frac{1}{\cosh(2\pi s)}*_se^{1/8\kappa^2}e^{-\frac{2s
    ^2}{\kappa^2}}\sin\left(\frac{s}{\kappa^2}\right)$ and $ b_2(r) = \frac{2}{\sqrt{\pi}}e^{-4r^2}e^{-i2r}$ where $*_s$ is the convolution with respect to the argument $s$.

\smallskip

The second is the Lindbladian from \cite{scandi2025thermalizationopenmanybodysystems} . It is given by $\cL_\alpha^{\operatorname{MB}} = \sum_a \cL_{a,\alpha}^{\operatorname{MB}}$ with (see Appendix~\ref{app:physical} for more details)

\begin{align}\label{eq:KMSMB}
		\cL^{\operatorname{MB}}_\alpha 
		&=-i\sqrbra{B_\alpha,\bigcdot\;}+\sum_a \int_{\mathbb{R}}\de\omega\;\Big({\widehat A}^a(\omega)(\, \bigcdot \,) {\widehat A}^{a\dagger}(\omega)-\frac{1}{2}\{{\widehat A}^{a\dagger}(\omega){\widehat A}^a(\omega),\bigcdot\;\}\Big)
	\end{align}
	where $B_\alpha$ is defined as in \eqref{eq:KMS-lindblad-B}, and the jump operators are
	\begin{align}
		\widehat{A}^a_\alpha(\omega) = \int_{\mathbb{R}} f_{a,\alpha}(t)  A^a(t) e^{-i \omega t } \text{d}t = \sum_{\nu\in B_H}\widehat{f}_{a,\alpha}(\omega-\nu)A^a_{\nu} ,\label{eq:jumpOperatorsDB}
	\end{align}
where now $f_{a,\alpha}(t)=\left ( \frac{1}{\sqrt{\sqrt{\pi} T(\alpha)}}\,\norbra{e^{\frac{-t^2}{2 T(\alpha)^2}}  \ast_t g_{a}}(-t) \right )$.
The corresponding coefficients in frequency as in \eqref{eq:KMSLind} for each term $\cL^{\operatorname{MB}}_{a,\alpha}$ are: 
\begin{align}\label{eq:KMS-coeff-MB-app-A} \Lambda^{{\operatorname{MB}},a,\alpha}_{\nu_1,\nu_2}= e^{-(T(\alpha)\,(\nu_1-\nu_2))^2/4}\;{\widehat{g}}_{a}\norbra{-\nu_1 } {\widehat{g}}_{a}\norbra{-{\nu_2} }.
\end{align}
where $\widehat{g}_a(\nu) = \widehat{g}_a(-\nu)e^{-\frac{\beta\nu}{2}}$, and where $\widehat{g}_a(\omega) \equiv \sqrt{\widehat{C}_{aa}(\omega)}$ such that $C_{aa}(t)=\int_\infty^\infty \operatorname{d}t\, g_a(s) g_a(t-s)$. Here $T(\alpha)=\frac{1}{2 \alpha \Gamma}\sqrt{2+3  \Gamma \tau}\propto \alpha^{-1}$ is called the observation time. A dimensionless version given by $T(\alpha)/\beta$ is the analog of $\kappa$ from \eqref{eq:CKG-lindbladian} for $\cL_{\alpha}^{\operatorname{MB}}$.

\smallskip 
The parameterizations in terms of $\kappa$ and $T(\alpha)/\beta$ are useful as they illuminate how the Lindbladians from \cite{chen2023efficient} and \cite{scandi2025thermalizationopenmanybodysystems} approach the Davies generator as $\kappa$ and $T(\alpha)/\beta$ increase. The dissipative parts of Davies generators are given by:
\begin{align}\label{eq:Davies-gen-form}
    \Lm^{\operatorname{D}} = \sum_a\sum_{\nu\in B_H}\gamma^{\operatorname{D}}(\nu)\left( A^a_\nu(\,\bigcdot\,)A_\nu^{a\dagger} -\frac{1}{2}\left\{A_\nu^{a\dagger} A^a_\nu,\bigcdot \,\right\}\right) ,
\end{align}
where the function $\gamma^{\operatorname{D}}$ satisfies the  KMS condition $\gamma^{\operatorname{D}}(-\nu) = e^{\beta\nu}\gamma^{\operatorname{D}}(\nu)$ (a concept distinct from that of KMS symmetry). Davies generators satisfy a stronger notion of detailed balance called Gelfand-Naimark-Segal (GNS) detailed balance, that imposes the stricter constraints on their general form evident from \eqref{eq:Davies-gen-form}; notably, they do not have off diagonal Bohr frequency components. It can be seen that in the limit $\kappa \to \infty$ and $T(\alpha)/\beta\to \infty$ that both $\cL_\kappa^G$ and $\cL^{\operatorname{MB}}_\alpha$  both become Davies generators, since 
\begin{align*}
 \Lambda^{\operatorname{G},\kappa}_{\nu_1,\nu_2}\to \delta_{\nu_1,\nu_2}e^{-\frac{(\beta\nu_1+1)^2}{4}}\qquad \text{ and }\qquad \Lambda^{{\operatorname{MB}},a,\alpha}_{\nu_1,\nu_2} \to \delta_{\nu_1,\nu_2}|\widehat{g}_a(\nu_1)|^2,
\end{align*}
where both $e^{-\frac{(\beta\nu+1)^2}{4}}$  and $|\widehat{g}_a(\nu)|^2$ can be seen to satisfy the KMS condition. For $\cL_\kappa^G$ in the intermediate regime where $\kappa \in [1,\infty)$, increasing $\kappa$ serves to more finely isolate energy transitions $\nu\in B_H$ by reducing the width of the frequency domain function $\widehat{f}_\kappa$. This comes at the cost of more Hamiltonian simulation resources required for its implementation on quantum computers and increased spatial extent of the jump operators \eqref{eq:jump-operators-kappa-app}, both stemming from the increased width of the time domain version of $f_\kappa$. The function $\gamma_\kappa(\omega)$ can also be seen from \eqref{eq:KMS-gamma-app} to approximately satisfy the KMS condition, with the approximation getting better as $\kappa$ increases. The parameter $T(\alpha)/\beta$ plays a similar role for $\cL^{\operatorname{MB}}_\alpha$, increasing both the energy transition resolution and the spatial extent of the underlying jump operators as it increases. With these two families of Lindbladians defined we now turn to proving Lemma~\ref{lem:CKG-monotonicity}, which shows that their spectral gaps and MLSI constants  increase monotonically as they approach the Davies generator limit.

\begin{proof}[Proof of Lemma~\ref{lem:CKG-monotonicity}.]
We begin by defining a general Dirichlet form that captures the families of Lindbladians in the lemma statement: 
\begin{align}\label{eq:dirichlet-form-delta}
    \Em_{\delta}(X,Y)=C(\delta)\sum_{a,\nu_1,\nu_2}e^{-\frac{(\beta\nu_1-\beta\nu_2)^2}{8\delta^2}}\frac{G(\nu_1,\nu_2) e^{-\beta(\nu_1+\nu_2)/4}}{2\cosh(\beta(\nu_1-\nu_2)/4)}\langle [A^a_{\nu_1},X],[A^a_{\nu_2},Y]\rangle_{\rho_\beta},
\end{align}   
where $G(\nu_1,\nu_2)=\overline{G}(-\nu_1,-\nu_2)e^{-\frac{\beta (\nu_1+\nu_2)}{2}}$, independent of $\delta$. The main examples are the following:
\begin{itemize}
    \item For the Lindbladians $\Lm^{\operatorname{G}}_\kappa$ \eqref{eq:CKG-lindbladian}, we first make the substitution $\delta = \frac{1}{\kappa}$ then choose $C(\delta) = \sqrt{\frac{2-\delta^2}{2}}$ and $G(\nu_1,\nu_2)=\exp\left(-\frac{(\beta\nu_1 +\beta\nu_2+2)^2}{16} \right)$. 
    
    \item For the Lindbladians $\cL_\alpha^{\operatorname{MB}}$ \eqref{eq:KMSMB} considered in \cite{scandi2025thermalizationopenmanybodysystems}
     we first make the substitution $\delta = \frac{\beta}{T(\alpha)}$ and then by \eqref{eq:KMS-coeff-MB-app-A}, $C(\delta) = 1$ and $G(\nu_1,\nu_2)=\widehat{g_a}(\nu_1)\widehat{g_a}(\nu_2)$ such that $\widehat{g_a}(\nu_1)=\widehat{g_a}(-\nu_1)e^{-\frac{\beta (\nu_1)}{2}}$ is defined through the Fourier transform of $C_{aa}(t)$ as $\widehat{g}_a(\omega) \equiv \sqrt{\widehat{C}_{aa}(\omega)}$ such that $C_{aa}(t)=\int_\infty^\infty \operatorname{d}s\, g_a(s) g_a(t-s)$. 
\end{itemize}
We next show that a function $\Em^{(\delta)}_{XY}$ corresponding to the modified Lindbladian $\widetilde{\Lm}_{\delta} = C(\delta)^{-1}\Lm_{\delta}$ (where $\cL_\delta$ is the Lindbladian corresponding to \eqref{eq:dirichlet-form-delta}) and given by:
\begin{align*}
    \Em^{(\delta)}_{XY}(s)=\sum_{a,\nu_1,\nu_2}e^{-\frac{(\beta\nu_1-\beta\nu_2)^2}{8\delta^2}}\frac{G(\nu_1,\nu_2) e^{-\beta(\nu_1+\nu_2)/4}}{2\cosh(\beta(\nu_1-\nu_2)/4)}\langle [A^a_{\nu_1},X(s\beta)],[A^a_{\nu_2},Y(s\beta)]\rangle_{\rho_\beta}
\end{align*}
satisfies the requirements \eqref{eq:convolution} and \eqref{eq:semigroupprop} of Lemma~\ref{lem:MLSI-gap} for suitable functions $f_\delta$ and $g_{\delta,\delta'}$. Indeed, by Fourier integration we have that 
\begin{align*}
    \Em^{(\delta)}_{XY}(s) = &\delta\sqrt{\frac{2}{\pi}}\int_\mathbb{R} e^{-2\delta^2r^2}\sum_{a,\nu_1,\nu_2}\frac{G(\nu_1,\nu_2) e^{-\beta(\nu_1+\nu_2)/4}}{2\cosh(\beta(\nu_1-\nu_2)/4)}\langle [A^a_{\nu_1},X((s-r)\beta)],[A^a_{\nu_2},Y((s-r)\beta)]\rangle_{\rho_\beta}dr\\
    =& \left(f_\delta\star\Em^{(\infty)}_{XY}\right)(s)
\end{align*}
with $f_\delta(s) = \delta \sqrt{\frac{2}{\pi}}e^{2\delta^2r^2}$ and where
\begin{align*}
    \Em^{(\infty)}_{XY}(s)= \sum_{a,\nu_1,\nu_2}\frac{G(\nu_1,\nu_2) e^{-\beta(\nu_1+\nu_2)/4}}{2\cosh(\beta(\nu_1-\nu_2)/4)}\langle [A^a_{\nu_1},X((s-r)\beta)],[A^a_{\nu_2},Y((s-r)\beta)]\rangle_{\rho_\beta}
\end{align*}
is the $\delta \to \infty$ limit of $\Em^{(\delta)}_{XY}$. It can then be verified that with the choice:
\begin{align*}
    g_{\delta,\delta'}(s) = \frac{\delta\delta'}{\sqrt{(\delta')^2-\delta^2}}\sqrt{2}{\pi}e^{-2\frac{\delta^2(\delta')^2}{(\delta')^2-\delta^2}s^2}
\end{align*}
for $ \delta\leq\delta' $ that $g_{\delta,\delta'}\star f_{\delta'}=f_\delta$. From Lemma~\ref{lem:MLSI-gap}, and from the fact that $\|g_{\delta,\delta'}\|_1=1$, it follows that:
\begin{align*}
   \lambda_{\textup{gap}}\left( \widetilde{\Lm}_{\delta}\right)\geq\lambda_{\textup{gap}}\left( \widetilde{\Lm}_{\delta'}\right)
\end{align*}
whenever $\delta\leq\delta'$. Then by noting that for $\delta \in[0,1]$ we clearly have that $\lambda_{\operatorname{Gap}}\left(\widetilde{\Lm}_\delta\right) = C(\delta)^{-1}\lambda_{\operatorname{Gap}}\left( \Lm_\delta\right) $ for both of the $C(\delta)$ functions considered here. It then follows that:
\begin{align*}
    \lambda_{\operatorname{Gap}}\left( \Lm_\delta\right)\geq \left(\frac{C(\delta)}{C(\delta')}\right)\lambda_{\operatorname{Gap}}\left( \Lm_{\delta'}\right)
\end{align*}
$\delta\leq\delta'$. For the Lindbladians of $\cL_\kappa^G$, upon resubstituting $\kappa = \frac{1}{\delta}$ we then have that for all $1\leq\kappa\leq \kappa'\leq\infty$:
\begin{align}\label{eq:gap-ineq-w-C-delta}
    \left(\frac{2-\frac{1}{\kappa'^2}}{2-\frac{1}{\kappa^2}}\right)^{1/2}\lambda_{\operatorname{Gap}}\left(\Lm^{\operatorname{G}}_{\kappa} \right)\leq \lambda_{\operatorname{Gap}}\left(\Lm^{\operatorname{G}}_{\kappa'} \right)
\end{align}
where $1\leq\left(\frac{2-\frac{1}{\kappa'^2}}{2-\frac{1}{\kappa^2}}\right)^{1/2}$ clearly for all $\kappa\leq\kappa'$. Then for $\cL^{\operatorname{MB}}_{\alpha}$ upon resubstituting $T(\alpha)= \frac{\beta}{\delta}$ and using the fact that $T(\alpha) \propto \alpha^{-1},$ we have that for all $0\leq\alpha \leq\alpha'$:
\begin{align} \label{eq:prooflemmamb} \lambda_{\text{gap}}\left(\cL^{\operatorname{MB}}_{\alpha'}\right)\leq\lambda_{\text{gap}}\left(\cL^{\operatorname{MB}}_{\alpha}\right).
\end{align}
For the MLSI constant, a similar proof yields the statements 
 \begin{align*}    
 \left(\frac{2-\frac{1}{\kappa'^2}}{2-\frac{1}{\kappa^2}}\right)^{1/2}\alpha_{\operatorname{MLSI}}\left(\Lm^{\operatorname{G}}_{\kappa} \right)\leq \alpha_{\operatorname{MLSI}}\left(\Lm^{\operatorname{G}}_{\kappa'} \right) \qquad \text{and}\qquad \alpha_{\operatorname{MLSI}}\left(\cL^{\operatorname{MB}}_{\alpha'}\right)\leq\alpha_{\operatorname{MLSI}}\left(\cL^{\operatorname{MB}}_{\alpha}\right)
 \end{align*}
 for $1\leq\kappa\leq \kappa'\leq\infty$ and $0\leq\alpha \leq\alpha'$ respectively.

\end{proof}

\paragraph{Warm up: Fast mixing of purely dissipative Lindbladians.} Before proving the results presented in the main text we present an illustrative example of how Lemma~\ref{lem:CKG-monotonicity} can be used to quickly prove fast mixing results for Lindbladians out of reach of previous proof methods when there is a spectral gap lower bound for a corresponding $\cL_1^{\operatorname{G}}$ Lindbladian. We consider an extensively normalized variant of the Lindbladian introduced in \cite{chen2023quantumthermalstatepreparation} (which corresponds to dissipative part of the Lindbladian \eqref{eq:CKG-lindbladian}) given by:
\begin{align*}
    \cL_{\kappa}^{\operatorname{C}} = ^{}\sum_{a}\int_{-\infty}^{\infty} \gamma_\kappa(\omega)\left( \widehat{A}_\kappa^a(\omega)(\,\bigcdot\,)\widehat{A}_\kappa^a(\omega)^\dagger -\frac{1}{2} \left\{\widehat{A}_\kappa^a(\omega)^\dagger \widehat{A}_\kappa^a(\omega),\bigcdot\; \right\}\right)d\omega.
\end{align*}
As this Lindbladian only differs from $\cL_\kappa^G$ by the coherent term we have that:
\begin{align*}
    \left\|\cL_\kappa^{\operatorname{C}}-\cL_\kappa^{\operatorname{G}}\right\|_{1-1}\leq 2\left\|B_\kappa^{\operatorname{G}}\right\|
    = &\cO\left(\|b_1\|_{L^1}\|b_2\|_{L^1}\sum_a{||A^a||^2}\right)\\ =& \cO\left(\frac{1}{\kappa}\left\| \cosh(2\pi s) \right\|_{L^1}\left\| e^{-\frac{2s^2}{\kappa^2}}\sin\left(\frac{s}{\kappa^2}\right)\right\|_{L^1} \|\cA\|_{G}\right)\\
    =&\cO\left(\frac{\|\cA\|_{G}}{\kappa}\right)
\end{align*}
where $\|\cA\|_{G} = \|\sum_a A^{a\dagger}A^a\|$. Now, the goal is to show we can control both the fixed point error and the mixing time of $\cL_\kappa^{\operatorname{C}}$ simultaneously. 
To show this we state the following Lemma without proof, noting that it is similar to Theorem 8 in \cite{ding2025endtoendefficientquantumthermal}.
\begin{lem}\label{lem:mixing-time-relations-L}
    Let $\cL_1$ and $\cL_2$ be primitive Lindbladians with fixed points $\rho_1$ and $\rho_2$ and mixing times $t_{\operatorname{mix},1}$ and $t_{\operatorname{mix},2}$ respectively. Then we have the following:
    \begin{enumerate}
        \item $\|\rho_1-\rho_2\|_1\leq \epsilon + t_{\operatorname{mix},1}(\epsilon)\|\mathcal L_1-\mathcal L_2\|_{1-1}$
        \item and $t_{\operatorname{mix},1}(\epsilon/4) \|\mathcal L_1-\mathcal L_2\|_{1-1}\leq \frac{\epsilon}{4}\implies t_{\operatorname{mix},2}(\epsilon) \leq t_{\operatorname{mix},1}(\epsilon/4) $.
    \end{enumerate}
\end{lem}
\noindent With this, since $\rho_\beta^S$ is the exact fixed point of $\cL_\kappa^G$, by setting $\cL_1 = \cL_\kappa^{\operatorname{G}}$ and $\cL_2 = \cL_\kappa^{\operatorname{C}}$ we have that 
\begin{align*}
    \|\rho_{\operatorname{fix}}(\cL_\kappa^{\operatorname{C}}) -  \rho_\beta^S\|_1 \leq \epsilon + \mathcal O\left(t_{\operatorname{mix},\cL_\kappa^{\operatorname{G}}}(\epsilon/4)\frac{\|\cA\|_{\operatorname{G}}}{\kappa}\right).
\end{align*}
where $\rho_{\operatorname{fix}}(\cL_\kappa^{\operatorname{C}})$ is the fixed point of $\cL_\kappa^{\operatorname{C}}$. By taking $\kappa = \Omega\left( \frac{t_{\operatorname{mix},\cL_\kappa^{\operatorname{G}}}(\epsilon/4)\|\cA\|_{\operatorname{G}}}{\epsilon} \right)$ we can ensure that both $\|\rho_{\operatorname{fix}}(\cL_\kappa^{\operatorname{C}}) -  \rho_\beta^S\|_1= \cO(\epsilon)$ and that $t_{\operatorname{mix},\cL_{\kappa}^{\operatorname{C}}}(\epsilon)\leq t_{\operatorname{mix},\cL_{\kappa}^{\operatorname{G}}}(\epsilon/4)$ simultaneously. Now assume that there exists a lower bound $\lambda_{\operatorname{Gap}}\left(\cL_1^{\operatorname{G}}\right)\geq \lambda$ on the spectral gap of $\cL_1^{\operatorname{G}}$. Then by equation~\eqref{eq:gap-mixing-time-ineq} and lemmas~\ref{lem:mixing-time-relations-L} and~\ref{lem:CKG-monotonicity} we have that:
\begin{align*}
    t_{\operatorname{mix},\cL_{\kappa}^{\operatorname{C}}}(\epsilon) \leq t_{\operatorname{mix},\cL_{\kappa}^{\operatorname{G}}}(\epsilon/4) \leq \frac{1}{\lambda_{\operatorname{Gap}}(\cL_\kappa^{\operatorname{G}})}\log\left(\frac{8\|\rho_\beta^{-1}\|}{\epsilon}\right)\leq \frac{1}{\lambda_{\operatorname{Gap}}(\cL_1^{\operatorname{G}})}\log\left(\frac{8\|\rho_\beta^{-1}\|}{\epsilon}\right)= \cO\left(\frac{1}{\lambda}\log\left(\frac{\|\rho_\beta^{-1}\|}{\epsilon}\right)\right).
\end{align*}

\section{Fast mixing of repeated interaction models.}\label{app:RIM}

\noindent In this section we consider the \emph{repeated interaction Gibbs sampling} algorithm (a version of the algorithm introduced in  \cite{ding2025endtoendefficientquantumthermal}) based on repeated interactions with a single qubit bath, and prove the fast mixing results presented in the main text.  One step of the algorithm is implemented by coupling a randomly chosen system operator $A^a$ (we take these in what follows to be single site Pauli operators or Majoranas depending on the context, and have $a$ index both the site and the operator type) to a single qubit bath in the thermal state $\rho_\beta^B = \frac{e^{-\beta H_B}}{\operatorname{Tr}[e^{-\beta H_B}]}$ with a Hamiltonian $H_B = -\frac{\omega}{2}Z$ at a randomly sampled frequency $\omega$. For a given $A^a$ and $\omega$ the time-dependent Hamiltonian evolution of the system and bath is given by:
\begin{align*}
    H_{(\alpha,\kappa)}(t) = H_S + \alpha f_\kappa(t)A^a \otimes X + H_B
\end{align*}
 where we take the function $f_\kappa$ as the Gaussian:
\begin{align*}
    f_\kappa(t) = \frac{1}{(\pi/2)^{1/4}\kappa^{1/2}\beta^{1/2}}e^{-\frac{t^2}{\kappa^2\beta^2}}
\end{align*}
with its width set by the dimensionless factor $\kappa$. The algorithm is then repeated applications of the channel:
\begin{align}\label{eq:end-to-end-channel-app}
    \Phi_{(\alpha,\kappa)}=\mathbb E_{A^a,\omega}\left[\operatorname{Tr}_B\left[U_{(\alpha,\kappa)}(T,-T)(\;\bigcdot\,\otimes \rho_\beta^B)U_{(\alpha,\kappa)}^\dagger(T,-T)\right]\right]
\end{align}
where $U_{(\alpha,\kappa)}(T,-T) =\mathcal T \exp\left(-i\int_{-T}^TH_{(\alpha,\kappa)}(s)ds\right)$. The expectation over the jump operators $A^a$ is taken uniformly over the set of jump operators and their negatives $\mathcal A :=\{A^a,-A^a\}_a$, while $\omega$ is sampled from the probability density:
\begin{align}\label{eq:g-omega-distribution-app}
    g(\omega) = \frac{\beta}{\sqrt{2\pi\left(2-\frac{1}{\kappa^2}\right)}}\exp\left( -\frac{(\beta\omega + 1)^2}{2\left(2-\frac{1}{\kappa^2}\right)}\right).
\end{align}
We now recall some of the results of \cite{ding2025endtoendefficientquantumthermal} adapted to the setting considered here. Explicitly, our setting corresponds to setting $\sigma = \kappa/2$, taking the set of $\mathcal A$ as all single site Pauli operators (or single site Majoranas) and their negatives, and setting $g(\omega)$ as in    \eqref{eq:g-omega-distribution-app}. In the context of \cite{ding2025endtoendefficientquantumthermal}'s Theorem 27, we additionally set their $x$ to $\frac{1}{\beta}$. The first shows that the dynamics of \eqref{eq:end-to-end-channel-app} is well approximated by a Lindbladian. 
\begin{thm}
[Adapted from Theorem 7 and Lemma 11 of \cite{ding2025endtoendefficientquantumthermal}]\label{thm:channel-to-lindbladian-error-bound}
    The channel \eqref{eq:end-to-end-channel-app} satisfies: \begin{align}\label{eq:quartic-channel-error-bound}
    \|\Phi_{(\alpha,
    \kappa)} - \U_S(T)\circ\exp(\alpha^2\Lm^{\operatorname{RI}}_\kappa)\circ\U_S(T)\|_{1-1} = \cO\left(\alpha^4T^4/(\kappa\beta)^2 + \alpha^2\kappa\beta \exp(-T^2/(\kappa\beta)^2)\right)
\end{align}
where $\mathcal U_S(T)= e^{-iH_S T}(\,\bigcdot\,) e^{iH_S T}$ and $\Lm^{\operatorname{RI}}_\kappa$ is a Lindbladian of the form:
\begin{align}\label{eq:approximating-lindbladian-k}
    \Lm^{\operatorname{RI}}_\kappa = \mathbb{E}_{A^a}\left[\int_{\mathbbm R}\left(-i\big[g(\omega) H^{\operatorname{LS}}_{A^a,f_\kappa} (\omega),\bigcdot\;\big] + \gamma(\omega)\cD_{\widehat{A}^a_{\kappa}(\omega)}\right)d\omega \right]
\end{align}
where $H^{\operatorname{LS}}_{A^a,f_\kappa}(\omega)$ is a Hermitian matrix, $\gamma(\omega) = (g(\omega)+g(-\omega))/(1+\exp(\beta\omega))$, and
\begin{align*}
    \cD_{\widehat{A}^a_{\kappa}(\omega)} = \widehat{A}^a_{\kappa}(\omega)(\,\bigcdot\,) \widehat{A}^a_{\kappa}(\omega)^\dagger -\frac{1}{2}\left\{\widehat{A}^a_{\kappa}(\omega)^\dagger \widehat{A}^a_{\kappa}(\omega),\bigcdot\;\right\} 
\end{align*}
with $\widehat{A}^a_{\kappa}(\omega)$ as in \eqref{eq:jump-operators-kappa-app}.
\end{thm}
\noindent This follows from the statements in \cite{ding2025endtoendefficientquantumthermal} by noting that $||f_\kappa||_{L^{\infty}}^4 = 1/(\kappa\beta)^2 $ and that for single site Paulis and Majoranas the expectation over the norms of the jump operators is upper bounded by $1$. The factor of $\alpha^2\kappa\beta \exp(-T^2/(\kappa\beta)^2$ appears in this result and not in Theorem 7 of \cite{ding2025endtoendefficientquantumthermal} because we choose to complete the integrals from $-T,T$ to $-\infty,\infty$ (as is done in Lemma 11 of \cite{ding2025endtoendefficientquantumthermal}) in order to make the correspondence with $\cL_\kappa^G$ more apparent. The next result gives the conditions for which for the fixed point of \eqref{eq:end-to-end-channel-app} is close the the thermal state $\rho_\beta^S = \frac{e^{-\beta H_S}}{\text{Tr}[e^{-\beta H_S}]}$.

\begin{thm}[Adapted from Theorem 12 of \cite{ding2025endtoendefficientquantumthermal}] \label{thm:fixed-point-error}
    When $\kappa \geq 2$, we have:
    \begin{align*} 
    \|\rho_{\operatorname{fix}}(\Phi_{(\alpha,\kappa)})-\rho_\beta^S\|_1
    = \widetilde{\cO}\left(\left(\frac{\beta}{\kappa}\left(1 + \sqrt{\log(\kappa/2)} \right)+ \beta\kappa\exp(-T^2/(\kappa^2\beta^2))+ \alpha^2T^4 (\kappa \beta)^{-2}\right)\alpha^2t_{\operatorname{mix},\Phi_{(\alpha,\kappa)}}(\epsilon)+\epsilon\right)
    \end{align*}
    where $\rho_{\operatorname{fix}}(\Phi_{(\alpha,\kappa)})$ is the fixed point of the channel $\Phi_{(\alpha,\kappa)}$.
\end{thm}
\noindent This follows by making the appropriate substitutions, and by noting that the scaling of $R$ defined in equation (E5) of \cite{ding2025endtoendefficientquantumthermal} in this case is given by equation I4 of \cite{ding2025endtoendefficientquantumthermal}, which when $x=1/\beta$ scales as $1/\kappa$. Finally we have the key theorem from \cite{ding2025endtoendefficientquantumthermal} which enables our fast mixing results, which shows that the Lindbladians $\cL_{\kappa}^{\operatorname{RI}}$ are closely related to the KMS symmetric Lindbladians $\cL_\kappa^{\operatorname{G}}$.  
\begin{thm}[Adapted from Theorem 27 of \cite{ding2025endtoendefficientquantumthermal}]\label{thm:L-close-to-KMS}
    Define $\|\mathcal{A}\|_{\operatorname{G}}:=\|\sum_{a}A^{a\dagger} A^a\|$, let $\Lm_\kappa^{\operatorname{G}}$ be the Lindbladian \eqref{eq:CKG-lindbladian}, and let $\Lm^{\operatorname{RI}}_\kappa$ be the Lindbladian in \eqref{eq:approximating-lindbladian-k}. Then \begin{align*}
        \Big\|\Lm^{\operatorname{RI}}_\kappa - \left(-i[C,\bigcdot\;]+\frac{\beta}{\|\mathcal A\|_{\operatorname{G}}\sqrt{2\pi(2-1/\kappa^2)}}\Lm_{\kappa}^G\right)\Big\|_{1-1} = \cO\left(\frac{\beta}{\kappa}\right)
    \end{align*}  
   where $C$ is a Hermitian operator satisfying:
    \begin{align*}
        \Big\|\rho_\beta^{S\;-1/4}C\rho_\beta^{S\;1/4} - \rho_\beta^{S\;1/4}C\rho_\beta^{S\;-1/4}\Big\| = \cO\left( \frac{\beta}{\kappa} \right).
    \end{align*}
\end{thm}
\noindent This follows by inserting the substitutions $\sigma = \kappa/2$ and $x=\frac{1}{\beta}$ and observing that these satisfy the requirements of the Theorem. The quantity $||\cA||_{\operatorname{G}}$ is defined to translate between the normalizations of the Lindbladians used here \eqref{eq:CKG-lindbladian} as well as in the constant spectral gap results \cite{rouzé2024efficientthermalizationuniversalquantum,tong2025fastmixingweaklyinteracting,bergamaschi2025quantumspinchainsthermalize} which have norms that are extensive in the system size, and the normalization used in \cite{chen2023efficient} (and referenced by Theorem 27 of \cite{ding2025endtoendefficientquantumthermal}) which have $\cO(1)$ norms.
The next two lemmas are also adapted from \cite{ding2025endtoendefficientquantumthermal} and will be used in the proof of Lemma~\ref{lem:gap-implies-efficient-prep-epsilon}.
\begin{lem}[Adapted from Corollary 26 from \cite{ding2025endtoendefficientquantumthermal}]\label{lem:lamb-shift-mixing-time-error-bound}
Let $\rho$ be a full rank state  and let $\Phi_\cM = \cU \circ \exp(\alpha^2\cM) \circ \cU$ where $\cM = -i[C,\bigcdot\;] + \Lm_{\textup{DB}}$ and $\cU$ is a unitary satisfying $\cU(\rho) = \rho$ and futhermore such that $\Lm_{\textup{DB}}$ satisfies KMS detailed balance with respect to $\rho$. Then whenever
\begin{align*}
    \|\rho^{-1/4}C\rho^{1/4}-\rho^{1/4}C\rho^{-1/4}\|\leq \xi\leq \lambda_{\textup{gap}}(\Lm_{\textup{DB}}) ,
\end{align*}
the mixing time of $\Phi_{\cM}$ satisfies:
\begin{align*}
    t_{\operatorname{mix},\Phi_{\cM}}(\epsilon)\leq \frac{1}{\alpha^2\big(\lambda_{\textup{gap}}(\Lm_{\textup{DB}})-\xi\big)}\log\left( \frac{\big\| \rho^{-1} \big\|}{\epsilon} \right)+1.
\end{align*}
\end{lem}
\noindent Note our $t_{\operatorname{mix},\Phi_{\cM}}$ is the standard mixing time as opposed to the $\alpha^2$ rescaled mixing time used in Corollary 26. This is why the version here has an $\alpha^{-2}$ factor.
\begin{lem}[Adapted from Theorem 12 from \cite{ding2025endtoendefficientquantumthermal}.]\label{lem:small-1-1-close-mixing-times}
    Let $\Phi_1$ and $\Phi_2$ be two CPTP maps with unique fixed points. Then given any $\epsilon>0$, if $t_{\operatorname{mix},\Phi_1}(\epsilon/4)\|\Phi_1-\Phi_2\|_{1-1}\leq \epsilon/4$, then 
    \begin{align*}
        t_{\operatorname{mix},\Phi_2}(\epsilon)\leq t_{\operatorname{mix},\Phi_1}(\epsilon/4).
    \end{align*}
\end{lem}
\noindent With these in hand, we can prove the primary lemma of this section.
\begin{proof}[Proof of Lemma~\ref{lem:gap-implies-efficient-prep-epsilon}]
 We begin by using Lemma~\ref{lem:CKG-monotonicity} to upper bound the mixing time of the channel 
 \begin{align*}
 \widehat{\Phi}_{(\alpha,\kappa)} = \cU_S(T)\circ \exp(\alpha^2\widehat\Lm_{\kappa})\circ\cU_S(T)
 \end{align*}
 where $\widehat\Lm_{\kappa}= -i[C,\bigcdot\;]+\frac{\beta}{\|\cA\|_{\operatorname{G}}\sqrt{2\pi(2-1/\kappa^2)}}\Lm_{\kappa}^G$, with $\|\cA\|_{\operatorname{G}}$ as defined in Theorem~\ref{thm:L-close-to-KMS}, in terms of $\lambda$. First, using Lemma~\ref{lem:CKG-monotonicity} $\big($we specifically use equation~\eqref{eq:gap-ineq-w-C-delta}$\big)$ and the assumed lower bound on the gap at $\kappa =1$ we have:
\begin{align}\label{eq:gap-delta-lower-bound-by-1}
    \lambda_{\textup{gap}}\left(\frac{\beta}{ \|\cA\|_{\operatorname{G}} \sqrt{2\pi(2-1/\kappa^2)}}\Lm_{\kappa}^G\right)\geq\frac{\beta}{\|\cA\|_{\operatorname{G}}}\frac{\sqrt{(2-1/\kappa^2)}}{\sqrt{2\pi(2-1/\kappa^2)}} \lambda_{\textup{gap}}\left(\Lm_{1}^G\right) \geq \frac{\beta}{\sqrt{2\pi}} \frac{\lambda}{\|\cA\|_{\operatorname{G}}}.
\end{align}
Next, noting that since in Theorem~\ref{thm:L-close-to-KMS} we have that $\Big\|\rho_\beta^{S\;-1/4}C\rho_\beta^{S\;1/4} - \rho_\beta^{S\;1/4}C\rho_\beta^{S\;-1/4}\Big\| = \cO\left( \frac{\beta}{\kappa}\right)$, we can choose a $
    \kappa^* = \Omega\big(\|\cA\|_{\operatorname{G}}/\lambda\big)$ such that for all $\kappa>\kappa^*$ 
    we have 
    \begin{align*}
    \Big\|\rho_\beta^{S\;-1/4}C\rho_\beta^{S\;1/4} - \rho_\beta^{S\;1/4}C\rho_\beta^{S\;-1/4}\|\leq \frac{1}{2}\frac{\beta}{\sqrt{2\pi}}\frac{\lambda}{\|\cA\|_{\operatorname{G}}}.
    \end{align*}
 Using this, we have by Lemma~\ref{lem:lamb-shift-mixing-time-error-bound}, equation~\eqref{eq:gap-delta-lower-bound-by-1}, (and assuming that $\alpha\leq 1$), that for all $\kappa\geq\kappa^*$:
\begin{align*}
    \widehat{\tau}_{\textup{mix}}(\epsilon/4):= \alpha^2t_{\operatorname{mix},\widehat{\Phi}_{(\alpha,\kappa)}}(\epsilon/4) \leq \frac{\sqrt{2\pi}\|\cA\|_{\operatorname{G}}}{\beta\lambda}\log\left(\frac{4\|\rho_\beta^{S\;-1}\|}{\epsilon}\right)+1
    = \cO\left( \frac{\|\cA\|_{\operatorname{G}}}{\beta \lambda} \log\left( \frac{\|\rho_{\beta}^{S\;-1}\|}{\epsilon} \right)\right) .
\end{align*}
    In the next stage of the proof, we need to use Lemma~\ref{lem:small-1-1-close-mixing-times} and a choice of parameters $\kappa,T,\alpha$ so we can ensure that $\alpha^2t_{\operatorname{mix},\Phi_{(\alpha,\kappa)}}(\epsilon)$ factor in Theorem~\ref{thm:fixed-point-error} can be replaced with $\widehat{\tau}_{\textup{mix}}(\epsilon/4)$. First note that by Theorem~\ref{thm:channel-to-lindbladian-error-bound} and Theorem~\eqref{thm:L-close-to-KMS}:
    \begin{align*}
        \Big\|\widehat{\Phi}_{(\alpha,\kappa)}  -  \Phi_{(\alpha,\kappa)}\big\|_{1-1}
        \leq  &\cO\left(\alpha^4T^4/(\kappa\beta)^2 + \alpha^2\kappa\beta\exp(-T^2/\kappa^2\beta^2)\right) + \alpha^2\|\Lm^{\operatorname{RI}}_\kappa -\widehat{\Lm}_\kappa\|_{1-1}\\
        = &\cO\left( \alpha^4T^4/(\kappa\beta)^2+\alpha^2\kappa\beta\exp(-T^2/\kappa^2\beta^2) + \frac{\alpha^2\beta}{\kappa}\right)
    \end{align*}
so that
\begin{align}\label{eq:mixing-time-times-1-1-norm-difference}
    t_{\operatorname{mix},\widehat{\Phi}_{\alpha,\kappa}}(\epsilon)\|\widehat{\Phi}_{\alpha,\kappa}-\Phi_{\alpha,\kappa}\|_{1-1} 
    = \cO\left(
    \left( \frac{\beta}{\kappa}+\kappa\beta\exp(-T^2/\kappa^2\beta^2)+\alpha^2T^4/(\kappa\beta)^2\right)\widehat{\tau}_{\textup{mix}}(\epsilon/4)\right).
\end{align}
Then by Lemma~\ref{lem:small-1-1-close-mixing-times} if \eqref{eq:mixing-time-times-1-1-norm-difference} is $\cO(\epsilon)$, we can exchange  the $\alpha^2t_{\operatorname{mix},\Phi_{(\alpha,\kappa)}}(\epsilon)$ with $\widehat{\tau}_{\textup{mix}}(\epsilon/4)$ in the error term in Theorem~\ref{thm:fixed-point-error}. 
 If we can then simultaneously ensure that:
\begin{align}\label{eq:fixed-point-error-eps-in-proof}
    &\left(\left(\frac{\beta
    }{\kappa}\left(1+\sqrt{\log(\kappa/2)}\right)+ \kappa\beta\exp(-T^2/(\kappa^2\beta^2))+ \alpha^2T^4 (\kappa \beta)^{-2}\right)\widehat{\tau}_{\textup{mix}}(\epsilon/4)\right)
\end{align}
from Theorem~\ref{thm:fixed-point-error} is $\cO(\epsilon)$ we can control the fixed point error as well. We now demonstrate the choices of $\kappa$, $T$, and $\alpha$ that ensure that \eqref{eq:mixing-time-times-1-1-norm-difference} and \eqref{eq:fixed-point-error-eps-in-proof} are both $\cO(\epsilon)$. First choose a $\kappa$ large enough such that it satisfies:
\begin{align*}
   \kappa =  \widetilde{\Omega}\left(\beta\epsilon^{-1}\widehat{\tau}_{\textup{mix}}(\epsilon/4)\right)
\end{align*}
so that the first term in both of the expressions \eqref{eq:mixing-time-times-1-1-norm-difference} and \eqref{eq:fixed-point-error-eps-in-proof} can be made small, while also scaling at least as fast as $\kappa^*$ since both $\epsilon^{-1}$ and $\log(\|\rho_{\beta}^{S\;-1}\|/\epsilon)$ are greater than $1$. We next must choose $T$ large enough so that the inverse exponential in the second term dominates the $\beta\kappa$ factors. This can be achieved by taking:
\begin{align*}
    T = \widetilde{\Omega}\left(  \kappa\beta \right).
\end{align*}
Next, we need to pick $\alpha$ small enough so that the third term can be made small. This can be done by choosing
\begin{align*}
    \alpha = \cO\left( T^{-2}\kappa\beta\epsilon^{1/2}\widehat{\tau}_{\textup{mix}}(\epsilon/4)^{-1/2} \right) = \widetilde{\cO}\left((\kappa\beta)^{-1}\epsilon^{1/2}\widehat{\tau}_{\textup{mix}}(\epsilon/4)^{-1/2}\right)
\end{align*}
The total Hamiltonian simulation time to produce a sample will be given by: $T\times t_{\operatorname{mix},\Phi_{(\alpha,\kappa)}}(\epsilon)$. Noting that by Lemma~\ref{lem:small-1-1-close-mixing-times}, it follows that $t_{\operatorname{mix},\Phi_{(\alpha,\kappa)}}(\epsilon/4)=\cO\left( \alpha^{-2} \widehat{\tau}_{\textup{mix}}(\epsilon/4)\right)$, and thus that:
\begin{align*}
  T\times t_{\operatorname{mix},\Phi_{(\alpha,\kappa)}}(\epsilon)=\cO\left(  \frac{\kappa\beta\widehat{\tau}_{\textup{mix}}(\epsilon/4)}{\alpha^2}\right)  = \widetilde{\cO}\left(\frac{\beta^6\widehat{\tau}_{\textup{mix}}(\epsilon/4)^5}{\epsilon^4}\right)  =\widetilde{\cO}\left( \frac{\beta\|\cA\|_{\operatorname{G}}^5}{\epsilon^4\lambda^5}\log\left( \|\rho_\beta^{S\;-1}\|\right)^5\right). 
\end{align*}
\end{proof}
\paragraph{Further details on Theorem~\ref{thm:repeated-interaction-efficient-prep}.}

In this section we expand on Theorem~\ref{thm:repeated-interaction-efficient-prep} stating and proving the results more formally. We begin with the statement for high temperature non-commuting qubit models given in \cite{rouzé2024efficientthermalizationuniversalquantum}. Given a lattice of qubits $\Lambda$, we consider systems $S$ of $(h,k,l)$-local Hamiltonians on a subregion $A\subset \Lambda$: 
\begin{align*}
   H_S\to H_A = \sum_{X\subseteq A} h_X
\end{align*}
where each non-zero $h_X$ acts non-trivially on at most $k$ qubits, each qubit has at most $l$ non-zero $h_X$'s acting non-trivially on it, and such that $\max_{X\subseteq \Lambda} \|h_X\| \leq h$. We will consider families of such Hamiltonians indexed by the size of their support $N\equiv|A|$, and where $h,k,$ and $l$ are constants independent of $N$.
\begin{prop}\label{prop:fast-prep-hkl}
    Let $H$ be an $(h,k,l)$-local lattice Hamiltonian on $N$ qubits as in the preceding paragraph, with associated thermal state at temperature $\beta$ 
    \begin{align*}
        \rho_\beta^S = \frac{e^{-\beta H}}{\textup{Tr}[e^{-\beta H}]}.
    \end{align*} Then there exists a temperature $\beta^*>0$ independent of the system size such that for any $\beta\leq \beta^*$ the repeated interaction Gibbs sampling algorithm prepares a state $\sigma$ satisfying $\|\sigma - \rho_\beta^S\|_1\leq\epsilon$ with total Hamiltonian simulation time:
    \begin{align*}
       t_{\textup{total}}(\epsilon) = \widetilde{\cO}\left( \frac{N^{10}}{\epsilon^4}\right).
    \end{align*}
    where $\widetilde{\cO}$ suppresses subleading polylogarithmic factors in $N$ and $\epsilon$.
\end{prop}
\begin{proof}
First note that Theorem II.1 in \cite{rouzé2024efficientthermalizationuniversalquantum} shows that given a Hamiltonian as in the corollary, there exists a $\beta^*=\cO((hkl)^{-1}) =\cO(1)$ such that the spectral gap of $\Lm^{\operatorname{G}}_1$ on this Hamiltonian for all $\beta\leq\beta^*$ is lower bounded by $\frac{1}{2\sqrt{2}e^{1/4}}=\Omega(1)$. The proposition then follows from Lemma~\ref{lem:gap-implies-efficient-prep-epsilon} and by noting that for $(h,k,l)$-local Hamiltonians: 
\begin{align*}
\log\left(\|\rho_\beta^{S\;-1}\| \right)\leq N+\beta hlN= \cO(N)
\end{align*}
and that here $\|\cA\|_{\operatorname{G}} = \cO(N)$ as $\cA$ is the set of all single site Pauli operators. 
\end{proof}
\smallskip
The next case we consider are models of weakly interacting fermions as described in \cite{tong2025fastmixingweaklyinteracting} where the jump operators in the Lindbladians $\cL_\kappa^{\operatorname{RI}}$ and $\cL_\kappa^{\operatorname{G}}$ are single-site Majoranas. There, the authors consider D-dimensional lattices of fermions with two Majorana fermion modes (corresponding to a single Dirac fermion) attached to each site satisfying the usual anti-commutation relations $\{\gamma_j,\gamma_k\}=2\delta_{jk}$. The Hamiltonians are such that $H_S= H_0 + V$ where $H_0= \sum_{j,k}h_{jk}\gamma_j\gamma_k$ is the non-interacting part and $V$ is a general parity preserving fermionic operator. It is further assumed that $H_0$ is $(1,r_0)$-geometrically local and $V$ is $(U,r_0)$-geometrically local, where $(\xi,r_0)$-geometrically locality is defined in the following definition:
\begin{Def}
    We say an operator $W$ is $(\xi,r_0)$-geometrically local if
    \begin{align*}
        W = \sum_{B\in\mathcal B(r_0)}W_B
    \end{align*}
    where $W_B$ only has non-trivial support on $B$, satisfies $\|W_B\|\leq\xi$, and $\mathcal B(r)$ denotes the set of all balls of radius $r$. 
\end{Def}
\noindent With this we can prove the following proposition.

\begin{prop}\label{prop:fast-prep-fermions}
    Let $H_S$ be a fermionic Hamiltonian over D-dimensional lattice on $N$ sites satisfying the assumptions of the preceding paragraph with associated thermal state at temperature $\beta$ given by $\rho_\beta^S$. Then for any fixed temperature $\beta\geq 0$, there exists a $U_\beta$ such that for all $U\leq U_\beta$, the repeated interaction Gibbs sampling algorithm prepares a state $\sigma$ satisfying $\|\sigma - \rho_\beta^S\|_1 \leq\epsilon$ with total Hamiltonian simulation time:
    \begin{align*}
    t_{\textup{total}}(\epsilon) = \widetilde{\cO}\left(\frac{N^{10}}{\epsilon^4}\right).
    \end{align*}
    where $\widetilde{\cO}$ suppresses subleading polylogarithmic factors in $N$ and $\epsilon$.
\end{prop}
\begin{proof}
    Theorem $3$ of \cite{tong2025fastmixingweaklyinteracting} implies that there exists a $U_\beta$ such that for all $U\leq U_\beta$, the spectral gap of $\Lm^{\operatorname{G}}_1$ restricted to the even subspace (see section 3 and Corollary 5 of \cite{tong2025fastmixingweaklyinteracting} for discussions on restricting to the even subspace) is lower bounded by a constant that depends only on $\beta,r_0$ and the dimension of the lattice. The result then follows from this combined with \ref{lem:gap-implies-efficient-prep-epsilon} (which has its conclusions unchanged when restricting to the even subspace), and by noting that for these models that 
    \begin{align*}
        \log\left(\|\rho_{\beta}^{S\;-1}\|\right)\leq N+\beta N = \cO(N)
    \end{align*}
    and $\|\cA\|_{\operatorname{G}} = \cO(N)$.
\end{proof}
\smallskip
We now turn to the case of 1-D models at any constant temperature as treated in \cite{bergamaschi2025quantumspinchainsthermalize}. On a $1$-D lattice with local dimension at each site given by $2^q$ we consider nearest neighbor Hamiltonians of the form 
\begin{align*}
    H_S=\sum_{b=1}^{n-1} H_{b,b+1}
\end{align*}
where $\|H_{b,b+1}\|\leq h$. 
\begin{prop}\label{prop:fast-prep-1D}
    Let $H_S$ be a local 1-$D$ lattice Hamiltonian as in the preceding paragraph on $n$ sites with associated thermal state $\rho_\beta^S$. Then for any fixed temperature $\beta\geq 0$ the repeated interaction Gibbs sampling algorithm prepares a state $\sigma$ satisfying $\|\sigma - \rho_\beta^S\|_1 \leq\epsilon$ with total Hamiltonian simulation time:
    \begin{align*}
        t_{\textup{total}}(\epsilon)=\widetilde{\cO}\left(\frac{N^{10}}{\epsilon^4}\right).
    \end{align*} 
    where $\widetilde{\cO}$ suppresses subleading polylogarithmic factors in $N$ and $\epsilon$.
\end{prop}
\begin{proof}
     Theorem I.1 of \cite{bergamaschi2025quantumspinchainsthermalize} shows that a system size independent lower bound on the gap of the $\Lm^{\operatorname{G}}_1$ Lindbladian exists that that depends only on $\beta$ and the local dimension $2^q$. The result follows from combining this with Lemma~\ref{lem:gap-implies-efficient-prep-epsilon}, and noting that for the models considered here we have that: 
    \begin{align*}
\log\left(\|\rho_{\beta}^{S\;-1}\| \right)\leq N+\beta N= \cO(N)
\end{align*}
and $\|\cA\|_{\operatorname{G}} = \cO(N)$.
\end{proof}
\noindent Combining propositions \ref{prop:fast-prep-hkl}, \ref{prop:fast-prep-fermions}, and \ref{prop:fast-prep-1D} then yields Theorem~\ref{thm:repeated-interaction-efficient-prep}.

\section{Fast Mixing of Physical Lindbladians }\label{app:physical}

\noindent We now consider a continuous evolution with a system-bath Hamiltonian $H=H_S+\alpha V + H_B$, with $V=\sum_a A^a \otimes B^a$, on an initial state $\rho_S(0)\otimes \rho_\beta^B$ such that $\rho_\beta^B=\frac{e^{-\beta H_B}}{Z_B}$, $Z_B=\tr{e^{-\beta H_B}}$, and $\norm{A^a}\le 1$ with $A^a$ all the single-body Paulis. This allows us to define the bath correlation function as $C_{ab}(t)=\tr{B^b(t)B^{a}\rho_\beta^B}-\tr{B^b\rho_\beta^B} \tr{B^{a}\rho_\beta^B}$, where for simplicity we shift all the operators $B^a$ such that $\tr{B^a\rho_\beta^B}=0$. We now explain how a KMS Lindbladian effectively emerges in this setting, following \cite{scandi2025thermalizationopenmanybodysystems}, with some extra assumptions needed for the efficiency proof.

Our first assumption is that the jump operators are uncorrelated by the bath dynamics \footnote{This is likely not strictly necessary, but greatly simplifies computations.}, such that 
$C_{ab}(t)=0$ if $a \neq b$. Considering its Fourier transform, we define $\widehat{g}_a(\omega) \equiv \sqrt{\widehat{C}_{aa}(\omega)}$ such that $C_{aa}(t)=\int_\infty^\infty \operatorname{d}t\, g_a(s) g_a(t-s)$, which allows us to define
\begin{align}
	\Gamma_a &:=  \norbra{\int_{-\infty}^\infty \vert g_a(t) \vert  \de{t} }^2 ,
	\\ 
	\Gamma_a \tau_a & :=\norbra{\int_{-\infty}^\infty \vert g_a(t) \vert  \de{t} } \norbra{\int_{-\infty}^\infty \vert t \vert\vert  g_a(t) \vert  \de{t} },
	\\ K_a & :=  \int_{-\infty}^\infty \de t \int_{-\infty}^\infty \de s \norbra{\vert t\vert + \vert s \vert }^2 \vert g_a(t) \vert \vert g_a(s) \vert ,
\end{align}
which we assume are all $\mathcal{O}\norbra{1}$. This also allows us to define $\sum_a \Gamma_a\le \Gamma$, and $\tau,K$ correspondingly. We also denote $\gamma_{\operatorname{max}}:=\max_a \Gamma_a$ and $\gamma_{\min}:=\min_a\Gamma_a$ in what comes next. Additionally, if $\Gamma_{0,a}=\int_{-\infty}^\infty \vert C_{aa}(t) \vert  \de{t} $, then $\vert \widehat C_{aa} (\omega) \vert \le \Gamma_{0,a} \le \Gamma_a$, and $\vert \widehat C_{aa}' (\omega) \vert \le \Gamma_a \tau_a$. With this, we are ready to state our main technical assumption on the bath.

\begin{Def} \label{def:bath}
	A heat bath with correlation function $ g(t)$ is \emph{mixing} if for all jump operators $A^a$ the following two conditions are satisfied:
	\begin{enumerate}
		\item for sufficiently small $\beta$, it holds that for all $a$:
		\begin{align}
			& \tau_a \le \mathcal{O}(\beta)\,;
		\end{align}
		\item let $h(t)$ be any monotonically increasing,  positive, even function, that diverges for $|t|\rightarrow\infty$ at most as $\Omega \norbra{\vert t \vert ^{2+2D}}$ in a $D-$dimensional system. Then, it holds that:
		\begin{align}
			&   	 \eta_h:=\frac{1}{\sqrt{\Gamma}}\int_{-\infty}^{\infty}\de t \; |g_{a}(t)| h(t)  < \infty\,.
		\end{align}
	\end{enumerate}
\end{Def}
\noindent Physically speaking, these assumptions mean that $1)$ in the high temperature limit the bath induces white noise, and $2)$ the bath has a short memory time, controlled by the decay of $h(t)^{-1}$.
As shown in Theorem 14 of \cite{scandi2025thermalizationopenmanybodysystems}, the evolution on the system is well approximated by a KMS symmetric  Lindbladian $\mathcal{L}^{\operatorname{MB}}_\alpha=\sum_a \mathcal{L}^{\operatorname{MB}}_{a,\alpha}$ of the form of \eqref{eq:KMSLind}, where each $\mathcal{L}^{\operatorname{MB}}_{a,\alpha}$ is defined through the coefficients
\begin{align} \label{eq:bathKMS}
	\Lambda^{{\operatorname{MB}},a}_{\nu_1,\nu_2}= e^{-(T(\alpha)\,(\nu_1-\nu_2))^2/4}\;{\widehat{g}}_{a}\norbra{-\nu_1 } {\widehat{g}}_{a}\norbra{-{\nu_2} },
\end{align} 
where $T(\alpha)$ is a free parameter called \emph{observation time} to be defined later. The system's frequencies are defined not with respect to $H_S$, but with respect to a renormalized Hamiltonian
$H_S^*=H_S+\alpha^2 H_{\operatorname{L}} $, where 
\begin{align}
	H_{\operatorname{L}} :=  \sum_{a} H_{\operatorname{L}}^{a} := i \sum_{a} \int_{-\infty}^\infty\de q \int_{-\infty}^{\infty}\de x\; \norbra{\frac{e^{-\frac{q^2+(x/2)^2}{T(\alpha)^2}}}{\sqrt{\pi}T(\alpha)}}C_{aa}(x){\rm sign}(x)\,A^{a}\norbra{q+\frac{x}{2}}A^a\norbra{q-\frac{x}{2}}\,.\label{eq:LSHam}
\end{align} 

\noindent The Lindbladian thus does not thermalize to $e^{-\beta H_S}/Z$, but to the Gibbs state of $H_S^*$. However, both Gibbs states are close in $1$-norm. In fact, as shown in \cite{scandi2025thermalizationopenmanybodysystems},
\begin{equation}\label{eq:distanceGibbs}
	\norm{\frac{e^{-\beta H_S}}{Z}-\frac{e^{-\beta H_S^*}}{Z^*}}_1\le  \alpha^2 \beta \norm{ H_{\operatorname{L}}} \le  \alpha^2 \beta \Gamma. 
\end{equation}
Additionally if $H_S$ is local, $H_S^*$ is still quasi-local with exponentially-decaying tails, as proven in the following:
\begin{lemma}\label{lemma:quasiLocHCG}
	This Lamb-shift Hamiltonian $ H_{\operatorname{L}}$ is such that one can define a Hermitian operator $H_{\operatorname{L}}^{r}$ which is sum of terms that have support at most of size $r$, so that:
	\begin{align}
		\norm {  H_{\operatorname{L}}^{a} -  H_{\operatorname{L}}^{a,r}} \le2\,\Gamma_a \norbra{\frac{(2e)^r}{\sqrt{2\pi r}}\exp\norbra{-\frac{r}{2}\log\norbra{\frac{r}{r_0}}}+4\,e^{-\frac{r}{r_0}}} \,,\label{eq:quasiLocalityEstimatesHCG}
	\end{align}
	where $r_0 :=\norbra{ v_{LR} T(\alpha)}$.
\end{lemma}

\begin{proof}
	Let us begin by introducing the modified evolution for the jump operators $A^{a,r}(t):=e^{iH_S^{r}t} A^a e^{-iH_S^{r}t}$, where $H_S^{r}$ is the system Hamiltonian truncated at a distance $r$ from the support of $A^a$. We are going to use the standard Lieb-Robinson bound:
	\begin{align}
		\Norm{A^{a}(t)-A^{a,r}(t)}\leq \frac{\norbra{v_{LR} |t|}^r}{r!}\,.\label{eq:LRlocal}
	\end{align}
	Let us define $H_{\operatorname{L}}^{r}$ in the same way as in Eq.~\eqref{eq:LSHam}, but with $A^{a,r}(t)$ instead of $A^{a}(t)$. We can bound the norm in Eq.~\eqref{eq:quasiLocalityEstimatesHCG} as:
	\begin{align}
		&\norm {  H_{\operatorname{L}}^a -  H_{\operatorname{L}}^{a,r}} \leq\nonumber
		\\
		&\leq\int_{-\infty}^{\infty}\de q\int_{-\infty}^{\infty}\de x \; \norbra{\frac{e^{-\frac{q^2+(x/2)^2}{T(\alpha)^2}}}{\sqrt{\pi}T(\alpha)}}|C_{aa}(x)|\,  \norbra{\Norm{A^{a} \norbra{q+\frac{x}{2}}-A^{a,r}\norbra{q+\frac{x}{2}}} + \Norm{A^{a} \norbra{q-\frac{x}{2}}-A^{a,r}\norbra{q+\frac{x}{2}}}}\,,
	\end{align}
	where we implicitly used $\|A^{a}(t)\|=\|A^{a,r}(t)\|\le 1$. To give an upper-bound, we separate the integration in three regions: that of $|q|\leq q_0$ and $|x|\leq x_0$ (for some arbitrary positive $q_0,\,x_0$); that of $|q|\leq q_0$ and $|x|\geq x_0$; and finally that of $|q|\geq q_0$. In the first region we can apply the bound in Eq.~\eqref{eq:LRlocal} to obtain:
	\begin{align}
		&\int_{-q_0}^{q_0}\de q\int_{-x_0}^{x_0}\de x \; \norbra{\frac{e^{-\frac{q^2+(x/2)^2}{T(\alpha)^2}}}{\sqrt{\pi}T(\alpha)}}|C_{aa}(x)|\,  \norbra{\Norm{A^{a} \norbra{q+\frac{x}{2}}-A^{a,r}\norbra{q+\frac{x}{2}}} + \Norm{A^{a} \norbra{q-\frac{x}{2}}-A^{a,r}\norbra{q+\frac{x}{2}}}}\leq
		\\
		&\leq2\int_{-q_0}^{q_0}\de q\int_{-x_0}^{x_0}\de x \; \norbra{\frac{e^{-\frac{q^2+(x/2)^2}{T(\alpha)^2}}}{\sqrt{\pi}T(\alpha)}}|C_{aa}(x)| \,\frac{v_{LR}^r}{r!}\norbra{|q| + \frac{|x|}{2}}^r \leq 2\,\Gamma_{0,a}\frac{1}{\sqrt{2\pi r}}\norbra{\frac{e v_{LR}}{\sqrt{r}}}^r\norbra{|q_0| + \frac{|x_0|}{2}}^r 
	\end{align}
	where in the second line we have taken the superior of the function in the last parenthesis, implicitly used the definition of $\Gamma_{0,a}$, and finally upper-bounded the factorial as $(r!)^{-1} \leq \frac{1}{\sqrt{2\pi r}}\norbra{\frac{e }{\sqrt{r}}}^r$.
	
	The contribution arising from the region $|q|\leq q_0$ and $|x|\geq x_0$ can be bounded using the trivial bound $\Norm{A^{a}(t)-A^{a,r}(t)}\leq 2 $. Then, we obtain:
	\begin{align}
		4\int_{-q_0}^{q_0}\de q\int_{|x|\geq x_0}\de x \;  \norbra{\frac{e^{-\frac{q^2+(x/2)^2}{T(\alpha)^2}}}{\sqrt{\pi}T(\alpha)}}|C_{aa}(x)|\leq 4\,e^{-\frac{(x_0/2)^2}{T(\alpha)^2}}\int_{-\infty}^\infty \de x \; |C_{aa}(x)| \leq 4\,\Gamma_{0,a}\,e^{-\frac{(x_0/2)^2}{T(\alpha)^2}}\,,  
	\end{align}
	where we integrated over $q$, took the maximum of the Gaussian in $x$ over the interval, and finally extended the integration limits. The last interval is $|q|\leq q_0$ and $|x|\geq x_0$. Once again, we use the trivial bound to get:
	\begin{align}
		4\int_{|q|\geq q_0}\de q\int_{-\infty}^\infty\de x \;  \norbra{\frac{e^{-\frac{q^2+(x/2)^2}{T(\alpha)^2}}}{\sqrt{\pi}T(\alpha)}}|C_{aa}(x)|\leq 4 \,\Gamma_{0,a}\, {\rm erfc}\norbra{\frac{q_0}{T(\alpha)}}\leq 4 \,\Gamma_{0,a}\, e^{-\frac{q_0^2}{T(\alpha)^2}}\,,
	\end{align}
	where in the last step we used the bound ${\rm erfc}\norbra{x}\leq e^{-x^2}$, valid for positive $x$.
	Adding everything together we have that:
	\begin{align}
		\norm {  H_{\operatorname{L}}^{a} -  H_{\operatorname{L}}^{a,r}} \leq 2\,\Gamma_{0,a}\norbra{\frac{1}{\sqrt{2\pi r}}\norbra{\frac{e v_{LR}}{\sqrt{r}}}^r\norbra{|q_0| + \frac{|x_0|}{2}}^r + 2 \norbra{e^{-\frac{(x_0/2)^2}{T(\alpha)^2}}+e^{-\frac{q_0^2}{T(\alpha)^2}}}}\,.
	\end{align}
	Setting $q_0 = x_0/2 = \sqrt{\frac{r T(\alpha)}{v_{LR}}}$, and using $ \Gamma_{0,a} \le \Gamma_a$, yields the claim.
\end{proof}

\noindent The Lindbladian defined in \eqref{eq:bathKMS} closely resembles the system's dynamics. This is the main result of \cite{scandi2025thermalizationopenmanybodysystems}, which bounds the smallest distance between the KMS Lindbladian and the real dynamics.

\begin{lemma}\label{le:approxKMS}
	The effective evolution induced on the system $\rho_S(t)= \operatorname{Tr}_B (e^{-i t \alpha^2 H} 
	(\rho_S(0) \otimes \rho_\beta^B) e^{i t \alpha^2 H})$ is such that,  
	\begin{align} \label{eq:closenessKMS}
		&\Norm{	\rho_S(t)-	\rho^{\operatorname{KMS}*}_S(t)}_1 
		\leq  \mathcal{O}\left({ \alpha^3 \,( \Gamma\, t)\Big(\Gamma\tau+(\Gamma\beta )\,e^{(\alpha\,\Gamma\beta)^2}\Big)}\right)\,.
	\end{align}
	where ${\rho}^{\operatorname{KMS}*}_S(t)=e^{t \alpha^2 \mathcal{L}^{\operatorname{MB}}_\alpha}(\rho_S(0))$ with $\mathcal{L}^{\operatorname{MB}}_\alpha$ the Lindbladian with coefficients as in \eqref{eq:bathKMS}, with $T(\alpha)=\frac{1}{2 \alpha \Gamma}\sqrt{2+3  \Gamma \tau}$.
\end{lemma}
This determines the optimal value of the free parameter $T(\alpha)$ at which, given a coupling constant $\alpha$, the KMS Lindbladian more closely resembles the real dynamics. 

In order to connect the gaps of this Lindbladian for different values of $T(\alpha)$, we use Lemma \ref{lem:CKG-monotonicity} as follows.
 Let us now consider $T(\alpha)$ as an arbitrary parameter, and define $\delta= \beta/T(\alpha)$ as above, so that in an evolution that is well approximated by the Lindbladian, $\delta \propto \beta \Gamma \alpha$ will be small. We can write the Dirichlet form of this Lindbladian (see App. \ref{app:KMSsymm}) as: 
\begin{align}
	\Em_{\delta}(X,Y)= \sum_{a,\nu_1,\nu_2}e^{-\frac{(\beta\nu_1-\beta\nu_2)^2}{8\delta^2}} {\widehat{g}}_{a}\norbra{-\nu_1 } {\widehat{g}}_{a}\norbra{-{\nu_2} }\frac{ e^{-\beta(\nu_1+\nu_2)/4}}{2\cosh(\beta(\nu_1-\nu_2)/4)}\langle [A^k_{\nu_1},X],[A^l_{\nu_2},Y]\rangle_{\rho_\beta}.
\end{align}

While for small $\delta$ the real evolution is close to a Lindbladian, we can only prove directly that the Lindbladian is gapped at $\delta'=1$. However, Lemma \ref{lem:CKG-monotonicity} allows us to connect the gap at both points.

In Sec. \ref{sec:delta1} we prove a lower bound of the gap at $\delta'=1$ in the case where the bath is mixing and its temperature $\beta^{-1}$ is larger than some constant depending on the system's parameters, while in Sec.~\ref{sec:delta2} we do the same by expanding around a weakly interacting Hamiltonian.

Given Lemma \ref{lem:CKG-monotonicity}, we then bound the gap for varying $\delta$ given the one at $\delta'=1$, so that, for $\delta \le \delta'$ (see \eqref{eq:prooflemmamb}),
\begin{align}\label{eq:boundsgaps}
	\lambda_\delta\geq \lambda_{\delta'}.
\end{align}
This allows us to establish, for $T(\alpha) \ge \beta$, that
$ \lambda_{\beta/T(\alpha)} \ge 
\lambda = \Omega(1)$, so we also lower bound the gap in the case of $T(\alpha)=\frac{1}{2 \alpha \Gamma}\sqrt{2+3 \Gamma \tau}$ (that is, the choice of observation time in Lemma \ref{le:approxKMS}). Moreover, it also holds that $\norm{\rho^{KMS*}_S(t) - \frac{e^{-\beta H_S^* }}{Z^*}}_1\le 2 \|\rho_\beta^{-1}\| e^{ - \lambda \alpha^2 t }$. Then, using \eqref{eq:distanceGibbs}, Lemma \ref{le:approxKMS} and the triangle inequality allows us to conclude that
\begin{equation}\label{eq:trianglebound}
	\norm{\rho_S(t) - \frac{e^{-\beta H_S }}{Z}}_1\le  2 \|\rho_\beta^{-1}\| e^{ - \lambda \alpha^2 t }+ \mathcal{O}\left({ \alpha^3\,(\Gamma\, t)\Big(\Gamma\tau+(\Gamma\beta )\,e^{(\alpha\,\Gamma\beta)^2}\Big)}\right)+ \alpha^2 (\Gamma\,\beta).
\end{equation}

\noindent By choosing $\alpha = \tilde{\mathcal{O}} \left ( \frac{\epsilon \lambda}{\log \norm{\rho_\beta^{-1}}\Gamma^2 (\tau+\beta)} \right)$ and resolving for the time $t$, we obtain Lemma \ref{lem:bathgap}. Note that $\lambda$ is the gap at $\delta'=1$, which corresponds to $\alpha= \frac{\sqrt{2+3 \Gamma \tau}}{2 \beta \Gamma}$ given the definition of $T(\alpha)$.

In order to prove Theorem \ref{thm:bath}, let us now consider that the number of jump operators is proportional to the system size, $3N$, so that $\Gamma,\tau = \mathcal{O} (N)$ and that $\|\rho_\beta^{-1}\|= e^{\mathcal{O}(N)}.$ Thus, choosing $ \alpha^2 t=\mathcal{O}(N)+\mathcal{O}(\log \epsilon^{-1})$ and $\alpha=\widetilde{\mathcal{O}}(N^{-3}\epsilon )$ we obtain $\norm{\rho_S(t) - \frac{e^{-\beta H_S }}{Z}}_1 \le \epsilon$ as desired. For the setting of perturbations around product Hamiltonians, the same proof also holds given the result in App. \ref{sec:delta2}.

\subsection{Gaps at $\delta'=1$ for high temperatures}\label{sec:delta1}

\noindent Building on the main result of~\cite{rouzé2024optimalquantumalgorithmgibbs}, in this section we prove that for $\delta'=1$ (that is $T(\alpha)= \beta$) the generator associated to the coefficients in  Eq.~\eqref{eq:bathKMS} (that is, the one in \eqref{eq:KMSMB}) is gapped at high temperatures. Let us label it as $\mathcal{L}^\beta$. 

To this end, we need to introduce some notation. First, we need to express the Lindbladian in the time picture. Following~\cite{scandi2025thermalizationopenmanybodysystems}, each term in the sum $\mathcal{L}^\beta=\sum_a \mathcal{L}^\beta_a$  reads: 
\begin{align}\label{eq:timeLind}
	\lind_a^{\beta} [\rho] 
	&=-\frac{i}{2}\sqrbra{H^{a,\beta}_{LS},\rho}+\ \int_{-\infty}^{\infty}\frac{\de\tau_1}{\sqrt{\pi}\, T(\alpha)} \int_{-\infty}^{\infty}\de \tau_2\int_{-\infty}^{\infty}\de\tau_3 \; e^{-\frac{\tau_1^2}{T(\alpha)^2}}{{g}}_{a}\norbra{\tau_2}{{g}}_{a}\norbra{\tau_3} \mathcal{D}^{a}_{(\tau_1+\tau_2),(\tau_1-\tau_3)}[\rho]
\end{align}
where we introduced the dissipator $\mathcal{D}^{a}$:
\begin{align}\label{eq:Dt1t2}
	\mathcal{D}^{a}_{t_1,t_2}[\rho] := \Big(A^a(t_1)\rho A^a(t_2)-\frac{1}{2}\{A^a(t_2)A^a(t_1),\rho\}\Big)\,,
\end{align}
and the Lamb-shift Hamiltonian:
\begin{equation}\label{eq:LSH}
	H^{a,\beta}_{LS} := \int_{-\infty}^{\infty}\frac{\de\tau_1}{\sqrt{\pi}\, T(\alpha)} \int_{-\infty}^{\infty}\de \tau_2\int_{-\infty}^{\infty}\de\tau_3 \;  F(\tau_1){{g}}_{a}\norbra{\tau_2}{{g}}_{a}\norbra{\tau_3} A^{a}(\tau_1-\tau_3)A^{a}(\tau_1+\tau_2)\,,
\end{equation}
where we introduced the function:
\begin{align}\label{eq:LSft}
	F(t) := \frac{2}{\beta}\int_{-\infty}^{\infty}\de x\; {\rm sech}\norbra{\frac{2\pi x}{\beta}}\sin\norbra{\frac{\beta (t-x)}{2 T(\alpha)^2}} e^{-\frac{(t-x)^2}{T(\alpha)^2}}e^{\frac{\beta^2}{16T(\alpha)^2}}\,.
\end{align}
It is important to keep in mind that the evolution of the jump operators is done with respect to the renormalized Hamiltonian
$H_S^*=H_S+\alpha^2 H_{\operatorname{L}}$. Additionally, we highlighted the dependency of the Lindbladian on the inverse temperature $\beta$, which enters through the correlation function of the bath ${{g}}_{a}\norbra{t}$. 

At this point, we need to introduce two additional ingredients. Let $A^{a,r}(t):=e^{i(H_S^*)^{r}t} A^a e^{-i(H_S^*)^{r}t}$ be a modified evolution of the jump operator $A^a$, where now $(H_S^*)^{r}$ only has support on a region of radius $r$ around the support of $A^a$. This makes $A^{a,r}(t)$ strictly local at all times. Then, we define $\mathcal{D}^{a,r}_{t_1,t_2}$ and $(H^{a,\beta}_{LS} )^r$ by replacing $A^a(t)$ in Eq.~\eqref{eq:Dt1t2} and Eq.~\eqref{eq:LSH} with $A^{a,r}(t)$. Finally, $\lind_a^{\beta,r}$ is defined by plugging $\mathcal{D}^{a,r}_{t_1,t_2}$ and $(H^{a,\beta}_{LS} )^r$ into Eq.~\eqref{eq:timeLind}. Studying $\lind_a^{\beta,r}$ will give us information about the locality structure of the full Lindbladian $\lind_a^{\beta}$. It should be noticed that, since $H_{\operatorname{L}}$ is quasi-local, the renormalized Hamiltonian $H_S^*$ satisfies a Lieb-Robinson bound such that
\begin{equation}
	\norm{ A^a(t) - e^{-it (H^*_{S})^r}A^a e^{it (H^*_{S})^r}}\le C e^{-\mu l}\left(e ^{v_{LR}^*t}-1 \right),\label{eq:LRquasilocal}
\end{equation}
with $C,\mu,v_{LR}^*$ are some constants generically depending on the geometry of the problem at hand, and on the details of $H_S^*$.

We also need to introduce the infinite temperature limit of the Lindbladian, which we denote by $\lind_a^{0}$. This reads:
\begin{align}
	\lind^{0}_a [\rho] 
	&=-\frac{i}{2}\sqrbra{H^{a,0}_{LS},\rho}+ \sum_a  \int_{-\infty}^{\infty}\frac{\de\tau_1}{\sqrt{\pi}\, T(\alpha)} \int_{-\infty}^{\infty}\de \tau_2\int_{-\infty}^{\infty}\de\tau_3 \; e^{-\frac{\tau_1^2}{T(\alpha)^2}}{{g}}_{a}\norbra{\tau_2}{{g}}_{a}\norbra{\tau_3} \mathcal{D}^{a}_{0,0}[\rho]\,,
\end{align}
where we introduced the infinite temperature Lamb-shift Hamiltonian:
\begin{align}
	H^{a,0}_{LS} :=\sum_a \int_{-\infty}^{\infty}\frac{\de\tau_1}{\sqrt{\pi}\, T(\alpha)} \int_{-\infty}^{\infty}\de \tau_2\int_{-\infty}^{\infty}\de\tau_3 \;  F(\tau_1){{g}}_{a}\norbra{\tau_2}{{g}}_{a}\norbra{\tau_3} A^{a}A^{a}\,.
\end{align}
If we take the jump operators as Paulis, this is trivially proportional to the identity. In this case $\mathcal{L}_{a}^{0}$ corresponds to the Lindbladian of a single-qubit depolarizing channel, which has rapid mixing with contant decay rate $\kappa_0$.

With this notations in hand, we can now prove that:
\begin{lemma}\label{lem:gapHT}
	For $\beta < \beta^*=\mathcal{O}(1)$ and a mixing bath, the gap of the Lindbladian in Eq.~\eqref{eq:timeLind} has a lower bound $\lambda_{1} \ge c $, where $c$ is a constant independent of system size.
\end{lemma}

The proof closely follows the one presented in~\cite{rouzé2024optimalquantumalgorithmgibbs} (see Appendices B.1 and B.2 in \cite{rouzé2024optimalquantumalgorithmgibbs}), and it follows from:

	\begin{lemma} \label{le:quasiLL}
	Given the Lindbladian in Eq.~\eqref{eq:timeLind}, one can prove the following two bounds:
	\begin{enumerate}
		\item Let $r_0$ be big enough so that $e^{-\frac{\mu r_0}{2}}\leq 1$, where $\mu$ is defined in Eq.~\eqref{eq:LRquasilocal}. It holds that:\label{it:quasiLL}
		\begin{align}
			\sum_{r\ge r_0}\|\mathcal{L}_{a}^{\beta,r\dagger}\!\!-\mathcal{L}_{a}^{\beta,r-1\dagger }\|_{\infty\to\infty}\!&\le \Delta(r_0)\,,\,,\label{eq:LemmalocLindbound}
		\end{align}
		with $\Delta(r_0) := \sum_{r \ge r_0} \phi(r)$, where, defining $r_0^*=v^*_{LR}T(\alpha)$,
		\begin{align}\label{eq:phiDef}
			\phi(r) :=& \norbra{4 + \frac{\beta e^{\frac{\beta^2}{16 T(\alpha)^2}}}{\sqrt{\pi} \,T(\alpha)}} \norbra{\frac{4 \Gamma_a(\eta_h+1)}{\left|h\norbra{\frac{\mu (r-1) T(\alpha)}{4\,r_0^*}}\right|+1}+4 \Gamma_a C e^{-\mu\frac{r-1}{2}}} \\ \nonumber &+16 \Gamma_a  e^{-\frac{\mu (r-1)}{4 r_0^*}}\norbra{ 1 + \frac{ \, T(\alpha)\,}{\beta} \,\norbra{1+e^{\frac{T^2}{4\beta^2}}e^{-(r-1)\frac{ \mu( \,T(\alpha)-\beta)}{4 \beta \,r_0^*}}}}\,.
		\end{align}
		\item The difference between the Lindbladian at a finite temperature and at $\beta=0$ is such that:\label{it:HT}
		\begin{align}
			\|\mathcal{L}_a^{\beta \dagger}-\mathcal{L}_{a}^{0 \dagger }\|_{\infty\to\infty}&\le \eta(\beta) \,,
		\end{align}
		where:
		\begin{align}
			\eta(\beta):= 2\,J\,\Gamma_a\norbra{\frac{2\, T(\alpha)}{\sqrt{\pi}} +4  \,\beta\norbra{e^{\frac{T(\alpha)^2}{4 \beta ^2}}+\norbra{\frac{T(\alpha)}{\beta}}^2}+ \norbra{1+\frac{\beta e^{\frac{\beta^2}{16 T(\alpha)^2}}}{4\sqrt{\pi} \,T(\alpha) }}\tau_a}\,,\label{eq:etaExpr}
		\end{align}
		and we introduced the constant $J= \mathcal{O}(1)$ satisfying $\|[H_S^*,A_a]\| \leq J$.
	\end{enumerate}
\end{lemma}

\begin{proof}[Proof of Lemma~\ref{lem:gapHT}] Suppose one could find two functions $\Delta(r_0)$ and $\eta(\beta)$ such that 
	\begin{enumerate}
		\item $\Delta(r_0)$ decays faster than $1/r_0^{2D+1}$ in a $D-$dimensional system;
		\item and $\eta(\beta)=\mathcal{O}(\beta)$;
	\end{enumerate}
and that satisfy the bounds:
\begin{align}
	\sum_{r\ge r_0}\|\mathcal{L}_a^{\beta,r \dagger}\!\!-\mathcal{L}_{a}^{\beta,r-1\dagger }\|_{\infty\to\infty}\!&\le \Delta(r_0)\,,\label{eq:locLindbound}\\
	\|\mathcal{L}_a^{\beta \, \dagger}-\mathcal{L}_{a}^{0 \, \dagger }\|_{\infty\to\infty}&\le \eta(\beta)\,.\label{eq:highTemperatureRed}
\end{align}
Then, in~\cite{rouzé2024optimalquantumalgorithmgibbs} it was shown that rapid mixing below a finite threshold temperature $\beta^*$ follows from the fact that  $\mathcal{L}_{a}^{0}$  has constant gap $\kappa_0$ (as it corresponds to a single-qubit depolarizing channel). 

The existence of such functions $\Delta(r_0)$ and $\eta(\beta)$ follows from Lemma~\ref{le:quasiLL}. Choosing $T(\alpha)=\beta$ and using the mixing property of the bath, we have
	\begin{align}
		\eta(\beta):= \mathcal{O}\norbra{ J \Gamma_a (\beta + \tau_a)} = \mathcal{O}(J \Gamma_a \beta),\label{eq:d34}
	\end{align}
	and $\Delta(r_0) = \sum_{r\ge r_0} \phi(r)$ where
	\begin{align}
		\phi(r) :=&  \mathcal{O}\norbra{\frac{\Gamma_a \eta_h}{\left|h\norbra{\frac{\mu (r-1) }{4\,v_{LR}^*}}\right|+1}} +  \mathcal{O}\norbra{ \Gamma_a e^{- \frac{\mu(r-1)}{4 \beta v_{LR}^* }}},
	\end{align}
	which decays as $h(t)^{-1}$ (which can be as fast as $\widehat g_a(t)$, considering Def. \ref{def:bath}) with an exponentially small correction.
	
	Using the result shown in appendix B.1 in \cite{rouzé2024optimalquantumalgorithmgibbs} , we have that the Lindbladian in Eq.~\eqref{eq:timeLind} has rapid mixing $\norm{\rho^{KMS*}_S(t) - \frac{e^{-\beta H_S^*}}{Z^*}}_1\le 4n e^{-(\kappa_0- \kappa) \alpha^2 t}$, where
	\begin{equation}
		\kappa = 4(2r_0+1)^{2D}\,\eta(\beta)\label{boundcrude}+L(r_0),
	\end{equation}
	where we define 
	\begin{equation}
		L(r_0)=5  (2r_0+1)^{2D} \Delta(r_0) +\left( 5+2r_0+2(2r_0+1)^{D} \right)\sum_{l \ge r_0} (2l+1)^{2D-1}\Delta(l) + 2 \sum_{l'\ge r_0} \sum_{l \ge l'} (2l+1)^{2D-2} \Delta(l).
	\end{equation}
	This function decays quickly with $r_0$ given the constraint on $h(t)$ in Def. \ref{def:bath}. Thus, we can choose $r_0,\beta$ in such a way that $L(r_0) < \frac{1}{4} \kappa_0$. This fixes $r_0$ as a constant, so that we can choose $\beta^*$ such that $4(2r_0+1)^{2D}\,\eta(\beta^*)<\frac{1}{4} \kappa_0$. This sets the expression for $\beta^*$, depending on $h(t)$ and constants $\mu, J , v_{LR}^*$ and both $\max_a \vert \Gamma_a \vert$ and $\max_a \vert \tau_a \vert$.
	
	Rapid mixing with a particular rate (in this case, $\frac{1}{2}\mu_0$)  implies that rate bounds the spectral gap by Lemma 6 in \cite{temme2015fast}. Thus, we conclude that $\lambda_1 \ge \frac{1}{2}\mu_0$ for $\beta <\beta^*=\mathcal{O} \norbra{1}$.
\end{proof}

We are left with proving Lemma~\ref{le:quasiLL}, which we do in the rest of the section: 
\begin{proof}[Proof of Lemma~\ref{le:quasiLL} (Item~\ref{it:quasiLL})] 
	Let us begin by noticing that the terms in Eq.~\eqref{eq:LemmalocLindbound} can be divided in dissipator and Lamb-shift rotation as:
	\begin{align}\label{eq:Linddist}
		\|\mathcal{L}_a^{\beta,r\dagger}\!\!-\mathcal{L}_{a}^{\beta,r-1\dagger }\|_{\infty\to\infty} \leq 2\, \|(H^{a,\beta}_{LS})^r-(H^{a,\beta}_{LS})^{r-1}\| + \|\mathcal{D}_a^{\beta,r\dagger}\!\!-\mathcal{D}_{a}^{\beta,r-1\dagger }\|_{\infty\to\infty}\,.
	\end{align}
	In what follows, we drop the index $a$ from the jump operators for simplicity of notation. We begin by analyzing the second term in the equation. To avoid cluttering the exposition, let us introduce:
	\begin{align}
		\mathcal{D}^{A,B}_{t_1,t_2}[\rho] := \Big(A(t_1)\rho B(t_2)-\frac{1}{2}\{B(t_2)A(t_1),\rho\}\Big)\,,
	\end{align}
	where $A$ and $B$ are two Hermitian operators. Then, a simple triangle inequality gives:
	\begin{align}
		\|(\mathcal{D}^{A^{r},A^{r}}_{t_1,t_2})^\dagger - (\mathcal{D}^{A^{r-1},A^{r-1}}_{t_1,t_2})^\dagger \|_{\infty\to\infty} \leq \|(\mathcal{D}^{A^{r},A^{r}}_{t_1,t_2})^\dagger - (\mathcal{D}^{A^{r},A^{r-1}}_{t_1,t_2})^\dagger \|_{\infty\to\infty}+\|(\mathcal{D}^{A^{r},A^{r-1}}_{t_1,t_2})^\dagger - (\mathcal{D}^{A^{r-1},A^{r-1}}_{t_1,t_2})^\dagger \|_{\infty\to\infty}\,,\label{eq:tria29}
	\end{align}
	where the single term can be bounded as:
	\begin{align}  \nonumber
		\|(\mathcal{D}^{A^{r},A^{r}}_{t_1,t_2})^\dagger - (\mathcal{D}^{A^{r},A^{r-1}}_{t_1,t_2})^\dagger \|_{\infty\to\infty} :&= \sup_{\|X\|=1} \, \Norm{A^r(t_1)X (A^r(t_2)-A^{r-1}(t_2))-\frac{1}{2}\{A^r(t_1)(A^r(t_2)-A^{r-1}(t_2)),X\}}\leq
		\\
		&\leq\,2\|(A^r(t_2)-A^{r-1}(t_2))\|\,,
	\end{align}
	and similarly for the second term in Eq.~\eqref{eq:tria29}. This gives us:
	\begin{align}
		& \|\mathcal{D}_a^{\beta,r\dagger}\!\!-\mathcal{D}_{a}^{\beta,r-1\dagger }\|_{\infty\to\infty} \noindent 
		\\
		&\leq 2
		\ \int_{-\infty}^{\infty}\frac{\de\tau_1}{\sqrt{\pi}\, T(\alpha)} \int_{-\infty}^{\infty}\de \tau_2\int_{-\infty}^{\infty}\de\tau_3 \; e^{-\frac{\tau_1^2}{T(\alpha)^2}}|{{g}}_{a}\norbra{\tau_2}\|{{g}}_{a}\norbra{\tau_3}| \bigg(\|(A^r(\tau_1+\tau_2)-A^{r-1}(\tau_1+\tau_2))\|+\nonumber
		\\
		&\qquad\qquad\qquad\qquad\qquad\qquad\qquad\qquad\qquad\qquad\qquad\qquad\qquad\qquad\qquad+\|(A^r(\tau_1-\tau_3)-A^{r-1}(\tau_1-\tau_3))\|\bigg)\,.\label{eq:boundDissLoc}
	\end{align}
	We begin by analyzing the first line (the second can be done similarly). First, we can carry out the integral over $\tau_3$, which gives a contribution of $\sqrt{\Gamma_a}$. Then, let us change variables $\tau_1 = x$ and $\tau_2 = q$, and introduce the two positive constants $q_0$ and $x_0$. In analogy with Lemma~\ref{lemma:quasiLocHCG}, we divide the integration in three regions: the region in which $|q|\leq q_0$ and $|x|\leq x_0$ (for some arbitrary positive $q_0,\,x_0$); the region in which $|q|\leq q_0$ and $|x|\geq x_0$; and finally the region in which $|q|\geq q_0$. 
	
	The integral over the first region can be bounded using the Lieb-Robinson bound in Eq.~\eqref{eq:LRquasilocal}, that is:
	\begin{align}
		2C e^{-\mu (r-1)}\int_{-x_0}^{x_0}\frac{\de x}{\sqrt{\pi}\, T(\alpha)}\int_{-q_0}^{q_0}\de q\; e^{-\frac{x^2}{T(\alpha)^2}}|{{g}}_{a}\norbra{q}||e ^{v_{LR}^*|q+x|}-1| \leq 2C \sqrt{\Gamma_a}\,e^{-\mu (r-1)}|e ^{v_{LR}^*(|q_0|+|x_0)|}-1|\,.
	\end{align}
	The initial factor $2$ arises from comparing both $A^r(\tau_1+\tau_2)$ and $A^{r-1}(\tau_1+\tau_2)$ with $A(\tau_1+\tau_2)$.
	On the other hand, the integral over the second region gives (using the trivial bound):
	\begin{align}
		4\int_{|x|\geq x_0}\frac{\de x}{\sqrt{\pi}\, T(\alpha)}\int_{-\infty}^{\infty}\de q\; e^{-\frac{x^2}{T(\alpha)^2}}|{{g}}_{a}\norbra{q}|\leq 4\sqrt{\Gamma_a}{\rm erfc}\norbra{\frac{x_0}{T(\alpha)}} \leq 4\sqrt{\Gamma_a} \,e^{-\frac{x_0^2}{T(\alpha)^2}}\,.
	\end{align}
	Finally, we also have:
	\begin{align}
		4\int_{-\infty}^\infty\frac{\de x}{\sqrt{\pi}\, T(\alpha)}\int_{|q|\geq q_0}\de q\; e^{-\frac{x^2}{T(\alpha)^2}}|{{g}}_{a}\norbra{q}|\norbra{\frac{|h(q)|+1}{|h(q)|+1}}\leq 4\sqrt{\Gamma_a}\norbra{\frac{\eta_h+1}{|h(q_0)|+1}}\,.
	\end{align}
	The same procedure also applies to the second term in Eq.~\eqref{eq:boundDissLoc}, giving exactly the same bounds. Summing everything up, choosing $q_0=x_0 = \frac{\mu (r-1)}{4v_{LR}^*}$, and assuming $r$ is big enough so that $e^{-\frac{\mu (r-1)}{2}}\leq 1$, we get:
	\begin{align}\label{eq:distD}
		\|\mathcal{D}_a^{\beta,r\dagger}\!\!-\mathcal{D}_{a}^{\beta,r-1\dagger }\|_{\infty\to\infty} \leq 16 \Gamma_a \norbra{ C e^{-\mu\frac{r-1}{2}}+e^{-\frac{\mu^2(r-1)^2}{16(r_0^*)^2}} + \frac{\eta_h+1}{\left|h\norbra{\frac{\mu (r-1) T(\alpha)}{4\,r_0^*}}\right|+1}}
	\end{align}
	where we have introduced $r_0^* := \norbra{ v_{LR}^* T(\alpha)}$.
	
	We now estimate the first term in Eq.~\eqref{eq:Linddist}, the norm-distance between the two Lamb-shift Hamiltonians. Using the same decomposition as in Eq.~\eqref{eq:LRlocal} we get:
	\begin{align}
		&\|(H^{a,\beta}_{LS})^r-(H^{a,\beta}_{LS})^{r-1}\|\leq
		\\
		&\leq \int_{-\infty}^{\infty}\frac{\de\tau_1}{\sqrt{\pi}\, T(\alpha)} \int_{-\infty}^{\infty}\de \tau_2\int_{-\infty}^{\infty}\de\tau_3 \;  |F(\tau_1)||{{g}}_{a}\norbra{\tau_2}||{{g}}_{a}\norbra{\tau_3}| \bigg(\|A^r_{}(\tau_1-\tau_3)-A^{r-1}_{}(\tau_1-\tau_3)\|+\nonumber
		\\
		&\qquad\qquad\qquad\qquad\qquad\qquad\qquad\qquad\qquad\qquad\qquad\qquad\qquad\qquad\qquad+ \|A^{r}_{}(\tau_1+\tau_2)-A^{r-1}_{}(\tau_1+\tau_2)\|\bigg)\,.
	\end{align}
	We now apply the same procedure as above. We first analyze the norm in the first line, and divide the integral in three parts. The first one gives:
	\begin{align}
		2C e^{-\mu (r-1)}\int_{-x_0}^{x_0}\frac{\de x}{\sqrt{\pi}\, T(\alpha)}\, \int_{-q_0}^{q_0}\de q\; |F(x)||{{g}}_{a}\norbra{q}||e ^{v_{LR}^*|q+x|}-1| \leq C \sqrt{\Gamma_a}\,e^{-\mu (r-1)}|e ^{v_{LR}^*(|q_0|+|x_0)|}-1| \frac{\beta e^{\frac{\beta^2}{16 T(\alpha)^2}}}{\sqrt{\pi} \,T(\alpha)}\,,
	\end{align}
	where we used the fact that:
	\begin{align}
		\int_{-\infty}^{\infty}\frac{\de x}{\sqrt{\pi}\, T(\alpha)} \;|F(x)| \leq  \frac{\beta e^{\frac{\beta^2}{16 T(\alpha)^2}}}{2\sqrt{\pi} \,T(\alpha)}\,.
	\end{align}
	For the next bound, we also need the fact that:
	\begin{align}
		|F(x)| \leq \frac{4\sqrt{\pi}\, T(\alpha)}{\beta} \,e^{-\frac{|x|}{T(\alpha)}}\norbra{1+e^{\frac{T(\alpha)^2}{4\beta^2}}e^{-|x|\frac{(\beta \,T(\alpha)-\beta^2)}{\beta^2 \,T(\alpha)}}}\,,\label{eq:uppBoundFX}
	\end{align}
	which can be verified by using ${\rm sech}(x)\leq e^{-|x|}$ and $\sin(x)\leq 1$ in Eq.~\eqref{eq:LSft}. This gives us the estimate:
	\begin{align}
		4\int_{|x|\geq x_0}\frac{\de x}{\sqrt{\pi}\, T(\alpha)}\int_{-\infty}^{\infty}\de q\;|F(x)||{{g}}_{a}\norbra{q}| \leq \frac{16\, T(\alpha)\sqrt{\Gamma_a}}{\beta} \,e^{-\frac{x_0}{T(\alpha)}}\norbra{1+e^{\frac{T^2}{4\beta^2}}e^{-x_0\frac{(\beta \,T(\alpha)-\beta^2)}{\beta^2 \,T(\alpha)}}}.
	\end{align}
	Finally, the last integral can also be bounded in a similar fashion:
	\begin{align}
		4\int_{-x_0}^{x_0}\frac{\de x}{\sqrt{\pi}\, T(\alpha)}\int_{|q|\geq q_0}\de q\;|F(x)||{{g}}_{a}\norbra{q}|\norbra{\frac{|h(q)|+1}{|h(q)|+1}}\leq \frac{2\sqrt{\Gamma_a}\,\beta e^{\frac{\beta^2}{16 T(\alpha)^2}}}{\sqrt{\pi} \,T(\alpha)}\norbra{\frac{\eta_h+1}{|h(q_0)|+1}}\,.
	\end{align}
	Putting everything together, choosing $q_0=x_0 = \frac{\mu (r-1)}{4v_{LR}^*}$, and assuming $e^{-\frac{\mu (r-1)}{2}}\leq 1$, we have:
	\begin{align} \label{eq:distH}
		\|(H^{a,\beta}_{LS})^r-(H^{a,\beta}_{LS})^{r-1}\|\leq\frac{4\Gamma_a\beta e^{\frac{\beta^2}{16 T(\alpha)^2}}}{\sqrt{\pi} \,T(\alpha)}\Bigg (& C e^{-\mu\frac{r-1}{2}}\,+ \norbra{\frac{\eta_h+1}{\left|h\norbra{\frac{\mu (r-1) T(\alpha)}{4\,r_0^*}}\right|+1}}+
		\\ \nonumber
		&\qquad\qquad+\frac{4\sqrt{\pi}\, T(\alpha)^2\,e^{-\frac{\beta^2}{16 T(\alpha)^2}}}{\beta^2} \,e^{-\frac{\mu (r-1)}{4 r_0^*}}\norbra{1+e^{\frac{T^2}{4\beta^2}}e^{-(r-1)\frac{ \mu( \,T(\alpha)-\beta)}{4 \beta \,r_0^*}}}\Bigg )\,.
	\end{align}
	Then by adding \eqref{eq:distD} and \eqref{eq:distH}, we obtain the expression for $\phi$ in~\eqref{eq:phiDef}.
	\end{proof}
	\begin{proof}[Proof of Lemma~\ref{le:quasiLL} (Item~\ref{it:HT})] 
		Again, we consider the dissipator and the unitary rotation separately. Following the same steps that led to Eq.~\eqref{eq:boundDissLoc}, we get:
		\begin{align}
			& \|\mathcal{D}_a^{\beta\dagger}\!\!-\mathcal{D}_{a}^{0\dagger }\|_{\infty\to\infty} \leq\nonumber
			\\
			&\leq 2
			\ \int_{-\infty}^{\infty}\frac{\de\tau_1}{\sqrt{\pi}\, T(\alpha)} \int_{-\infty}^{\infty}\de \tau_2\int_{-\infty}^{\infty}\de\tau_3 \; e^{-\frac{\tau_1^2}{T(\alpha)^2}}|{{g}}_{a}\norbra{\tau_2}||{{g}}_{a}\norbra{\tau_3}| \bigg(\|(A^a(\tau_1+\tau_2)-A^a)\|+\|(A^a(\tau_1-\tau_3)-A^a)\|\bigg)
		\end{align}
		We can bound the operator norm in the first line as:
		\begin{align}
			\|(A^a(\tau_1+\tau_2)-A^a)\|\leq (|\tau_1|+|\tau_2|)\|[H_S,A^a]\|\,,\label{eq:boundCommcose}
		\end{align}
		and similarly for the second line. Choosing a constant $J$ such that for all $A^a$, we have $\|[H_S^*,A^a]\| \leq J $, we carry out the integral to obtain:
		\begin{align}\label{eq:beta01}
			\|\mathcal{D}_a^{\beta\dagger}\!\!-\mathcal{D}_{a}^{0\dagger }\|_{\infty\to\infty} \leq 2\,J\,\Gamma_a\norbra{\frac{2\, T(\alpha)}{\sqrt{\pi}} + \tau_a}\,.
		\end{align}
		We have a similar behavior also for the unitary term, since we have that:
		\begin{align}
			&\|H^{a,\beta}_{LS}-H^{a,0}_{LS}\|\leq\nonumber
			\\
			&\leq \int_{-\infty}^{\infty}\frac{\de\tau_1}{\sqrt{\pi}\, T(\alpha)} \int_{-\infty}^{\infty}\de \tau_2\int_{-\infty}^{\infty}\de\tau_3 \;  |F(\tau_1)||{{g}}_{a}\norbra{\tau_2}||{{g}}_{a}\norbra{\tau_3}| \bigg(\|A^a(\tau_1-\tau_3)-A^a\|+ \|A^a(\tau_1+\tau_2)-A^a\|\bigg)\,,
		\end{align}
		so that we can apply Eq.~\eqref{eq:boundCommcose} and the upper bound in Eq.~\eqref{eq:uppBoundFX} to obtain:
		\begin{align}\label{eq:beta02}
			\|H^{a,\beta}_{LS}-H^{a,0}_{LS}\|\leq 2\,J\,\Gamma_a\norbra{4  \,\beta\norbra{e^{\frac{T(\alpha)^2}{4 \beta ^2}}+\norbra{\frac{T(\alpha)}{\beta}}^2}+\norbra{\frac{\beta e^{\frac{\beta^2}{16 T(\alpha)^2}}}{4\sqrt{\pi} \,T(\alpha) }} \, \tau_a}\,.
		\end{align}
		Adding Eq.~\eqref{eq:beta01} and~\eqref{eq:beta02} gives the claim in Eq.~\eqref{eq:etaExpr}. 
\end{proof}

\subsection{Gaps at $\delta'=1$ for weakly interacting Hamiltonians}\label{sec:delta2}

\noindent In~\cite{smid2025rapidmixingquantumgibbs}, it was shown that the same techniques of~\cite{rouzé2024optimalquantumalgorithmgibbs} can be applied to the following situation: consider a non-interacting Hamiltonian $H_S^0 := \sum_i h_i$, where each of the $h_i $ only acts on one component of the systems. Then, let $H_S^\lambda := H_S^0 + \lambda V_S$, where $V_S$ is an interaction term that can be decomposed as:
\begin{align}
	V_S = \sum_{r\geq 1} \sum_{C \in C(r)} W_C\,,\label{eq:intCharact}
\end{align}
where $C(r)$ is the set of balls of radius $r$, and the terms $W_C$ satisfy $\|W_C\|\leq K e^{-\nu r}$. Denote $\lind_a^{\lambda}$ the Lindbladian thermalizing to the Gibbs state of $(H_S^\lambda)^*$ (the renormalized interacting Hamiltonian). Moreover, let us also denote $\lind_a^{\lambda=0}$ the one thermalizing to $(H_S^0)^*$. Then, since $\lind_a^{\lambda=0}$ is gapped, one can use the same techniques of~\cite{rouzé2024optimalquantumalgorithmgibbs} to prove that there exists a  $\lambda^*$ sufficiently small (but independent of the system size), such that $\forall\lambda\leq \lambda^*$, also $\lind_a^{\lambda}$ is gapped.

The reason why this is possible appears evident once the main ingredients in the proof of Lemma~\ref{lem:gapHT} are laid out: first, one needs a Lindbladian for which one can prove a size-independent gap. In the high-temperature case this was given by the Lindbladian of the depolarizing channel, in this case, by $\lind_a^{\lambda=0}$. Second, one needs to prove that the Lindbladian is sufficiently local for all the parameters (see Lemma~\ref{le:quasiLL}, item~\ref{it:quasiLL}). Both in the proof of the high temperature regime and in here, this relies on the existence of Lieb-Robinson bounds for Hamiltonians with exponential tails. This also means that the proof of Lemma~\ref{le:quasiLL}, item~\ref{it:quasiLL} carries out here completely unchanged. Finally, the third ingredient is that small perturbation around the gapped Lindbladian only perturb it linearly. This corresponds to Eq.~\eqref{eq:d34}. This is the only ingredient that does not carry out from the previous section, which we prove here: 
\begin{lemma}
	The difference between the Lindbladian at a finite interaction and at $\lambda=0$ is such that:
	\begin{align}
		\|\mathcal{L}_a^{\lambda \dagger}-\mathcal{L}_{a}^{\lambda=0 \dagger }\|_{\infty\to\infty}&\le \eta(\lambda) \,,
	\end{align}
	where:
	\begin{align}
		\eta(\lambda):= 4 \,C_{D,\nu} K\lambda\,\Gamma_a\norbra{\frac{2\, T(\alpha)}{\sqrt{\pi}} +4  \,\beta\norbra{e^{\frac{T(\alpha)^2}{4 \beta ^2}}+\norbra{\frac{T(\alpha)}{\beta}}^2}+ \norbra{1+\frac{\beta e^{\frac{\beta^2}{16 T(\alpha)^2}}}{4\sqrt{\pi} \,T(\alpha) }}\tau_a}\,,\label{eq:etaExpr2}
	\end{align}
	where $K$ and $\nu$ are the constants appearing in Eq.~\eqref{eq:intCharact}, and $C_{D,\nu}$ is a constant depending on the dimension and on $\nu$.
\end{lemma}
\begin{proof}
	Let us define $A^{a,\lambda}(t) := e^{i(H_S^\lambda)^*t} A^a e^{-i(H_S^\lambda)^*t}$ and $A^{a,\lambda,r}(t) := e^{i((H_S^\lambda)^*)^rt} A^a e^{-i((H_S^\lambda)^*)^rt}$, where $((H_S^\lambda)^*)^r$ acts non-trivially only in a ball of radius $r$ around the support of $A^a$. Then, it holds that:
	\begin{align}
		\|A^{a,\lambda}(t) -A^{a,0}(t) \|\leq \|A^{a,\lambda}(t) -A^{a,\lambda,r}(t) \| +  \|A^{a,\lambda,r}(t) -A^{a,0,r}(t)\|+ \|A^{a,0,r}(t) -A^{a,0}(t)\|\,.\label{eq:d64}
	\end{align}
	We can bound the norm in the middle as:
	\begin{align}
		\|A^{a,\lambda,r}(t) -A^{a,0,r}(t)\| &\leq \int_{0}^{\lambda}\de x\int_{0}^{t}\de s\;\|[(V^r)^{x,r}(s),A^{a,x,r}(t)]\|\leq 2\lambda \,|t| \|A^a\| \|V_S^r\|\leq
		\\
		&\leq 2\lambda\,|t| \sum_{r\geq \ell\geq 1} \norbra{\frac{\pi^{D/2}}{\Gamma\norbra{\frac{D}{2}+1}} \ell^{D}} \, K e^{-\nu \ell}
	\end{align}
	where in the second line used an an upper-bound on the number of balls of radius $\ell$ around $A^a$ by the volume of the ball in $D$ dimensions, and used Eq.~\eqref{eq:intCharact}. Since the sum above converges as $r\rightarrow\infty$, we can give the upper-bound:
	\begin{align}
		\|A^{a,\lambda,r}(t) -A^{a,0,r}(t)\| \leq 2\lambda\,|t| \sum_{ \ell= 1}^\infty \norbra{\frac{\pi^{D/2}}{\Gamma\norbra{\frac{D}{2}+1}} \ell^{D}} \, K e^{-\nu \ell} = 2 \,C_{D,\nu} K\lambda |t|\,,
	\end{align}
	where we implicitly defined $C_{D,\nu}$ to be the value of the sum above. Since this bound holds irrespective of $r$, we can take the limit $r\rightarrow\infty$ in Eq.~\eqref{eq:d64} to also obtain:
	\begin{align}
			\|A^{a,\lambda}(t) -A^{a,0}(t) \|\leq 2 \,C_{D,\nu} K\lambda |t|\,.
	\end{align}
	This is all we need to prove our claim. Indeed, since the dependency on time of the norm above is the same as the one appearing in the proof of Lemma~\ref{le:quasiLL}.\ref{it:HT}, by repeating the same steps from there we get:
	\begin{align}
		\eta(\lambda):= 4 \,C_{D,\nu} K\lambda\,\Gamma_a\norbra{\frac{2\, T(\alpha)}{\sqrt{\pi}} +4  \,\beta\norbra{e^{\frac{T(\alpha)^2}{4 \beta ^2}}+\norbra{\frac{T(\alpha)}{\beta}}^2}+ \norbra{1+\frac{\beta e^{\frac{\beta^2}{16 T(\alpha)^2}}}{4\sqrt{\pi} \,T(\alpha) }}\tau_a}\,.
	\end{align}
	This proves the claim.
\end{proof}

For a mixing bath, by choosing $T(\alpha) = \beta$ as per Def. \ref{def:bath}, the error scales as $\eta(\lambda) = \bigo{\lambda \,\beta}$. This shows that the product between the inverse temperature and the interaction strength actually is the important figure of merit for our result. In particular, following the discussion above, we obtain:
\begin{lemma}\label{lem:gaplowInt}
	For $(\beta\lambda) < \xi^*=\mathcal{O}(1)$ and a Markovian bath, the gap of the Lindbladian in Eq.~\eqref{eq:timeLind} defined in terms of $H^\lambda_S$ (see Eq.~\eqref{eq:intCharact}) has a lower bound $\lambda_{1} \ge c $, where $c$ is a constant independent of system size.
\end{lemma}

\section{Faster convergence via reduction to the Davies Lindbladian in the commuting case}\label{app:davies}

\noindent Suppose $H_S$ is a $k$-local commuting Hamiltonian, so that it can be written as $H_S = \sum_{i}\; h_i$, where the sum is over all system sites, each $h_i$ acts over at most $k=\mathcal{O}(1)$ terms, and for any $i$ and $j$, it holds that $[h_i, \,h_j]=0$. For these systems, the commutativity of the $\{h_i\}$ implies that the jump operators are strictly local. This special feature allows us to approximate the physical evolution with a Davies Lindbladian \cite{Davies1974MarkovMasterI,Davies1976MarkovMasterII} with an error bound compatible with the other approximations, as  recalled here. The commutativity by itself also implies that the relevant minimum non-zero spacing between any two different frequencies $\nu_-^{\operatorname{min}}= \Omega \left (1 \right )$ is independent of the system size. 
Let the Davies generator be given in the frequency representation as: 
\begin{align}
	{\lind}^{\operatorname{D}} [\rho] 
	&=\sum_a \sum_{\nu\in B_H} \norbra{-\frac{i}{2}S_{a}^{\operatorname{D}}(\nu)\sqrbra{A^{a \dagger}_\nu A^a_\nu ,\rho}+\gamma_{a}^{\operatorname{D}}(\nu)\Big(A^a_\nu \rho A^{a\dagger}_\nu-\frac{1}{2}\{A^{a \dagger}_\nu A^{a}_\nu,\rho\}\Big)}\,.
\end{align} 
where the rates $\gamma_{a}^{\operatorname{D}}(\nu)$ and $S_{a}^{\operatorname{D}}(\nu)$ are related to the correlation function of the bath as:
\begin{align}
    \begin{cases}
        \gamma_{a}^{\operatorname{D}}(\nu) = \widehat{C}_{aa}\norbra{\nu}
        \\
        S_{a}^{\operatorname{D}}(\nu) = {\rm P.V.}\norbra{\frac{1}{\pi}\int_{-\infty}^{\infty}\de\omega^*\;\frac{\widehat{C}_{aa}(\omega^*)}{(\omega^*-\nu)}}
    \end{cases}\,.
\end{align}
In~\cite{scandi2025thermalizationopenmanybodysystems} it was shown that the rates of $\lind^{\operatorname{MB}}_\alpha$ are close to the ones of the Davies generator. This result, together with the locality of the jump operators, allows us to prove the following. 
\begin{lemma}\label{lem.Daviesapprox}Let $H_S$ be a $k$-local commuting Hamiltonian. Let $d_k$ be the dimension of the Hilbert space corresponding to a region of size $k$ in the physical space. Then, it holds that for any initial state $\rho_S(0)$:
\begin{align}
	&\| \rho_S(t)\!-\! e^{t \alpha^2 \mathcal{L}^{\operatorname{D}}}\rho_S(0)\|_1 
	\!=\!  \mathcal{O}\!\left(\alpha^3 (\Gamma t) (\Gamma\tau)\!\left( 1 \!+d_k^6  \right) \right)\!.
\end{align}
\end{lemma}
\begin{proof}
Following~\cite{scandi2025thermalizationopenmanybodysystems}, we introduce the coarse grained Lindbladian:
\begin{align}
	{\lind}_\alpha^{\operatorname{CG}} [\rho] 
	&=\sum_a \sum_{\nu_1,\nu_2\in B_H} \norbra{-\frac{i}{2}S_{a}^{\nu_1,\nu_2}\sqrbra{A^{a\dagger}_{\nu_2} A^a_{\nu_1},\rho}+\gamma_{a}^{\nu_1,\nu_2}\Big(A^a_{\nu_1} \rho A^{a\dagger}_{\nu_2} -\frac{1}{2}\{A^{a \dagger}_{\nu_2} A^a_{\nu_1},\rho\}\Big)}\,,
\end{align}
where we introduced the rates:
\begin{align}
	\begin{cases}
		\gamma_{a}^{\nu_1,\nu_2}=\frac{e^{-(T(\alpha)\,\nu_-)^2/4}}{\sqrt{\pi}}\int_{-\infty}^{\infty}\de\Omega\;{\widehat{C}}_{aa}\norbra{\frac{\Omega}{T(\alpha)}+\nu_+}\,e^{-\Omega^2}\;\vspace{0.2cm}
		\\
		S_{a}^{\nu_1,\nu_2}=\frac{e^{-(T(\alpha)\,\nu_-)^2/4}}{\sqrt{\pi}}\int_{-\infty}^{\infty}\frac{\de\Omega}{\pi}\int_{-\infty}^{\infty}\de\Omega' \;\frac{\widehat{C}_{aa}(\Omega')}{(\Omega'-\nu_+)-\frac{\Omega}{T(\alpha)}}\;e^{-\Omega^2} \,
	\end{cases}\,,\label{eq:ratesCG}
\end{align} 
where $\nu_+ = \frac{\nu_1+\nu_2}{2}$ and $\nu_- = \nu_1-\nu_2$.
In~\cite{scandi2025thermalizationopenmanybodysystems} it was proven that:
\begin{align}
    \| \rho_S(t)\!-\! e^{t \alpha^2 \mathcal{L}^{\operatorname{CG}}_\alpha}\rho_S(0)\|_1 
	\!=\!  \mathcal{O}\!\left(\alpha^3 (\Gamma t) (\Gamma\tau) \right)\,.
\end{align}
Then, it is sufficient to prove that $\|\mathcal{L}^{\operatorname{CG}}_\alpha-\lind^{\operatorname{D}}\|_{1-1}\leq \bigo{\alpha \Gamma^2\tau}$ to prove the claim (see the discussion in~\cite{scandi2025thermalizationopenmanybodysystems}). This result relies on two facts: first, the rates in Eq.~\eqref{eq:ratesCG} satisfy~\cite{scandi2025thermalizationopenmanybodysystems}:
\begin{align}
	\begin{cases}
		|\gamma_{a}^{\nu_1,\nu_2}  -\gamma_{a}^{\operatorname{D} }(\nu_1)| \leq \frac{\Gamma_a \tau_a}{\sqrt{\pi} T(\alpha)}& \qquad\qquad{\rm if}\; \nu_1=\nu_2 \\
		|\gamma_{a}^{\nu_1,\nu_2}|\leq e^{-(T(\alpha)\,\nu_-^{\operatorname{min}})^2/4}\,\Gamma_a\norbra{1+\frac{ \tau_a}{\sqrt{\pi} T(\alpha)}} & \qquad\qquad{\rm otherwise}
	\end{cases}\;;
	\\ 
	\begin{cases}
		|S_{a}^{\nu_1,\nu_2}-S_{a}^{\operatorname{D}}(\nu_1) |  \leq\frac{\Gamma_a \tau_a}{\sqrt{\pi} T(\alpha)}& \qquad\qquad{\rm if}\; \nu_1=\nu_2 \\
		|S_{a}^{\nu_1,\nu_2}|\leq e^{-(T(\alpha)\,\nu_-^{\operatorname{min}})^2/4}\,\Gamma_a\norbra{1+\frac{ \tau_a}{\sqrt{\pi} T(\alpha)}} & \qquad\qquad{\rm otherwise}
	\end{cases}\;,
\end{align}
where $\Gamma_a$ and $\tau_A$ satisfy $\Gamma = \sum_a \Gamma_a$ and $(\Gamma\,\tau) = \sum_a (\Gamma_a\tau_a)$. Second, the jump operators $A^a_{\nu}$ become strictly local (see \cite[Lemma 11]{kastoryano2016quantumgibbssamplerscommuting}).
Locality is key in order to get an upper-bound on $\|A_a(\omega)\|$ that is independent of $N$. Indeed, we have:
\begin{align}
	\|A^a_{\nu}\| = \norm{\sum_{E_i-E_j = \nu} \; \ket{E_i}\bra{E_i}A^a\ket{E_j}\bra{E_j}}\leq \sum_{E_i-E_j = \nu} \; \norm{\ket{E_i}\bra{E_i}A_a\ket{E_j}\bra{E_j}}\leq d^2_k
\end{align}
where $d_k$ is the dimension of the Hilbert space over a region of size $k$. 

We are now ready to prove the claim. Thanks to the locality of the jumps, these are at most $d_k^2$. Then, we have that: 
\begin{align*}
	\|{\lind}^{\operatorname{CG}} [\rho] -{\lind}^{\operatorname{D}} [\rho] \|_1&\leq \sum_{a,\,\omega}\norbra{\frac{3\Gamma_a \tau_a}{\sqrt{\pi}\, T(\alpha)}d_k^4} + \sum_{\substack{a,\,\omega,\,\widetilde{\omega}\\\omega\neq\widetilde{\omega}}}\norbra{3\,e^{-(T(\alpha)\,\nu_-^{\operatorname{min}})^2/4}\,\Gamma_a\norbra{1+\frac{ \tau_a}{\sqrt{\pi} \,T(\alpha)}}d_k^4}\\
    &\leq 3\,\Gamma d_k^6\norbra{\frac{\tau}{\sqrt{\pi} \,T(\alpha)}+e^{-(T(\alpha)\,\nu_-^{\operatorname{min}})^2/4}\norbra{1+\frac{ \tau}{\sqrt{\pi}\, T(\alpha)}}d_k^2}\,.
\end{align*}
It should be noted that $\nu_-^{\operatorname{min}}$ is independent of the system size. Then, choosing $T(\alpha) \propto\alpha^{-1}$, we obtain that, up to exponentially small corrections:
\begin{align}
	&\| \rho_S(t)\!-\! e^{t \alpha^2 \mathcal{L}^{\operatorname{D}}}\rho_S(0)\|_1 
	\!=\!  \mathcal{O}\!\left(\alpha^3 (\Gamma t) (\Gamma\tau)\!\left( 1 \!+d_k^6  \right) \right)\!.
\end{align}
This proves the claim. 
\end{proof}

\begin{cor}
Assuming that the Davies generator satisfies a modified logarithmic Sobolev inequality with constant $\alpha_{\operatorname{MLSI}}(\cL^{\operatorname{D}})$, for a time $t=\widetilde{\mathcal{O}}\Big(\frac{\gamma_{\max}^4N^4\tau^2}{\alpha_{\operatorname{MLSI}}(\cL^{\operatorname{D}})^3\epsilon^2}\Big)$ and a sufficiently weak coupling $\alpha=\widetilde{\Omega}\Big(\frac{\epsilon\,\alpha_{\operatorname{MLSI}}(\cL^{\operatorname{D}})}{\gamma_{\max}^2N^2\tau}\Big)$, the system's state $\rho_S(t)$ is $\epsilon$-close to the Gibbs state in $1$-norm.
\end{cor}

\begin{proof}
By Lemma \ref{lem.Daviesapprox}, we have that 
\begin{align}\label{eq:bound1too}
\|\rho_S(t)-e^{t\alpha^2\cL^{\operatorname{D}}}(\rho_S(0))\|_1=\mathcal{O}\left(\alpha^3 \Gamma^2 t \tau\right).
\end{align}
Moreover, by the modified logarithmic Sobolev inequality together with Pinsker's inequality, we get (cf.~\eqref{eq.MLSIMixing})
\begin{align*}
\|e^{t\alpha^2\cL^{\operatorname{D}}}(\rho_S(0))-\rho_\beta\|_1\le \sqrt{2 e^{-\alpha^2\alpha_{\operatorname{MLSI}}(\cL^{\operatorname{D}}) t}\log\|\rho_\beta^{-1}\|}
\end{align*}
In order for the first bound to be below $\epsilon$, we need to choose $$\alpha^2 t=\Theta\left(\frac{1}{\alpha_{\operatorname{MLSI}}(\cL^{\operatorname{D}})}\Big(\log\log\|\rho_\beta^{-1}\|+\log\frac{1}{\epsilon}\Big)\right).$$
Next, we ensure that the bound in \eqref{eq:bound1too} is within $\mathcal{O}(\epsilon)$ by choosing 
\begin{align*}
\alpha=\Theta\left(\frac{\epsilon\alpha_{\operatorname{MLSI}}(\cL^{\operatorname{D}})}{\Gamma^2\tau\Big(\log\log\|\rho_\beta^{-1}\|+\log\frac{1}{\epsilon}\Big)}\right)=\widetilde{\Omega}\left(\frac{\epsilon\,\alpha_{\operatorname{MLSI}}(\cL^{\operatorname{D}})}{\gamma_{\max}^2N^2\tau  }\right).
\end{align*}
This gives the bound 
\begin{align*}
t=\widetilde{\mathcal{O}}\left( \frac{\gamma_{\max}^4N^4\tau^2}{\alpha_{\operatorname{MLSI}}(\cL^{\operatorname{D}})^3\epsilon^2} \right).
\end{align*}

\end{proof}

\noindent The modified logarithmic Sobolev inequality constant $\alpha_{\operatorname{MLSI}}(\cL^{\operatorname{D}})$ has recently been investigated for a variety of models \cite{capel2020modified,kochanowski2024rapidthermalizationdissipativemanybody,bardet2023rapid}, including for CSS Hamiltonians \cite{stengele2025modified} and on one-dimensional lattices \cite{bardet2024entropy}, where it was shown to scale polylogarithmically with the number of sites.

\end{document}